\documentclass[11pt, a4paper]{article}

\usepackage{xr}
\usepackage{amsthm}
\usepackage{amsmath}
\usepackage{amsfonts}
\usepackage{booktabs}
\usepackage{natbib}
\usepackage{hyperref}
\usepackage{graphicx}
\usepackage{color}
\usepackage[left = 2.5cm, right = 2.5cm, bottom = 3cm, top = 2.5cm]{geometry}
\usepackage{multirow}
\usepackage{rotating}
\usepackage{calc}
\usepackage{bm}
\usepackage{pdfpages}
\usepackage{dcolumn}

\newcommand*{\bb}{\boldsymbol}

\newcommand{\Op}[1]{\ensuremath{{\mathcal{O}_p(#1)}}}

\newcommand{\vnorm}[1]{\ensuremath{{\left\| #1 \right\|}}}
\newcommand{\mnorm}[1]{{\left\vert\kern-0.25ex\left\vert\kern-0.25ex\left\vert #1 
		\right\vert\kern-0.25ex\right\vert\kern-0.25ex\right\vert}}
\newcommand{\mnorms}[1]{{\vert\kern-0.25ex\vert\kern-0.25ex\vert #1 
		\vert\kern-0.25ex\vert\kern-0.25ex\vert}}

\newtheoremstyle{example}
{3pt} 
{3pt} 
{} 
{0\parindent} 
{\bf}
{:} 
{.5em} 
{} 
\newtheoremstyle{theorem}
{3pt} 
{3pt} 
{\em} 
{0\parindent} 
{\bf}
{:} 
{.5em} 
{} 
\theoremstyle{example} 
\theoremstyle{theorem} \newtheorem{theorem}{Theorem}[section]

\def\det{{\mathop{\rm det}}}

\def\bbeta{\bb{\beta}}
\def\btheta{\bb{\theta}}
\def\bpsi{\bb{\psi}}
\def\bgamma{\bb{\gamma}}
\def\bSigma{\bb{\Sigma}}
\def\bgamma{\bb{\gamma}}
\def\by{\bb{y}}
\def\bC{\bb{C}}
\def\bY{\bb{Y}}
\def\bu{\bb{u}}
\def\bx{\bb{x}}
\def\bz{\bb{z}}
\def\b0{\bb{0}}
\def\bX{\bb{X}}
\def\bZ{\bb{Z}}
\def\bv{\bb{v}}
\def\bV{\bb{V}}
\def\bY{\bb{Y}}
\def\by{\bb{y}}
\def\bL{\bb{L}}

\def\btnod{\bb{\theta}_0}
\def\bttilde{\tilde{\bb{\theta}}}
\def\bW {\bb{W}}
\title{Maximum softly-penalized likelihood for mixed effects logistic regression}

\author{Philipp Sterzinger \\
	\texttt{philipp.sterzinger@warwick.ac.uk} \smallskip 
	\\ and \smallskip \\
	Ioannis Kosmidis \\ \texttt{ioannis.kosmidis@warwick.ac.uk} \bigskip \\
	Department of Statistics, University of Warwick\\
	Coventry CV4 7AL, UK
}

\begin{document}
	\maketitle
	
	\begin{abstract}
		Maximum likelihood estimation in logistic regression with mixed
		effects is known to often result in estimates on the boundary of the
		parameter space. Such estimates, which include infinite values for
		fixed effects and singular or infinite variance components, can
		cause havoc to numerical estimation procedures and inference. We
		introduce an appropriately scaled additive penalty to the
		log-likelihood function, or an approximation thereof, which
		penalizes the fixed effects by the Jeffreys' invariant prior for the
		model with no random effects and the variance components by a
		composition of negative Huber loss functions.  The resulting maximum
		penalized likelihood estimates are shown to lie in the interior of
		the parameter space. Appropriate scaling of the penalty guarantees
		that the penalization is soft enough to preserve the optimal
		asymptotic properties expected by the maximum likelihood estimator,
		namely consistency, asymptotic normality, and Cram\'{e}r-Rao
		efficiency. Our choice of penalties and scaling factor preserves
		equivariance of the fixed effects estimates under linear
		transformation of the model parameters, such as contrasts. Maximum
		softly-penalized likelihood is compared to competing approaches on
		two real-data examples, and through comprehensive simulation studies
		that illustrate its superior finite sample performance.
	\end{abstract}

	\section{Introduction}
	\label{sec:intro}
	Generalized Linear Mixed Models (GLMMs; \citealt[Chapter
	7]{mcculloch+etal:2008}) are a potent class of statistical models that
	can associate Gaussian and non-Gaussian responses, such as counts,
	proportions, positive responses, and so on, with covariates, while
	accounting for complex multivariate dependencies. This is achieved by
	linking the expectation of a response to a linear combination of
	covariates and parameters (fixed effects), and sources of extra
	variation (random effects) with known distributions. Among GLMMs,
	mixed effects logistic regression is arguably the most frequently used
	model to analyse binary response outcomes. Although these models
	find application in numerous fields such as biology, ecology and the
	social sciences \citep{bolker+etal:2009}, estimation of GLMMs is not
	straightforward in practice because their likelihood is generally an
	intractable integral.
	
	Maximum likelihood (ML) methods for GLMMs maximize the GLMM likelihood
	or an approximation thereof, which can in principle be made
	arbitrarily accurate \citep[see, for example,
	][]{raudenbush+etal:2000, pinheiro+chao:2006}. Such methods are
	pervasive in contemporary GLMM practice because the resulting
	estimators are consistent under general model regularity conditions,
	and the resulting estimates and the (approximate) likelihood can be
	used for likelihood-based inferential devices, such as
	likelihood-ratio or Wald tests, and model selection procedures based
	on information criteria. An alternative approach is to use Bayesian
	posterior update procedures \citep[see, for
	example,][]{zhao+etal:2006,browne+draper:2006}. However, such
	procedures come with various technical difficulties, such as
	determining the scaling of the covariates, selecting appropriate
	priors, coming up with efficient posterior sampling algorithms, and
	determining burn-in times of chains for reliable estimation. Yet
	another alternative are maximum penalized quasi-likelihood methods
	\citep{schall:1991, wolfinger+oconnel:1993, breslow+clayton:1993}
	which essentially fit a linear mixed model to transformed
	pseudo-responses. However, the penalized quasi-likelihood may not
	yield an accurate approximation of the GLMM likelihood. As a result,
	these estimators can have large bias when the random effects variances
	are large \citep{bolker+etal:2009,rodriguez+goldman1995} and are not
	necessarily consistent \citep[Chapter 3.1]{jiang:2017}.
	
	Despite the pervasiveness of ML methods in the statistical practice
	for GLMMs, certain data configurations can result in estimates of the
	variance-covariance matrix of the random effects distribution to be on
	the boundary of the parameter space, such as infinite or zero
	estimated variances, or, more generally, singular estimates of the
	variance-covariance matrix. In addition, as is the case in ML
	estimation for Bernoulli-response generalized linear models
	\citep[GLMs; see, for example][Chapter 4]{mccullagh+nelder:1989}, the
	ML estimates of Bernoulli-response GLMMs, such as the mixed effects
	logistic regression model, can be infinite. As is well-acknowledged in
	the GLMM literature \citep[see, for example][]{bolker+etal:2009,
		bolker:2015, pasch+etal:2013}, both instances of estimates on the boundary of the parameter space
	can cause havoc to numerical optimization procedures, and if such
	estimates go undetected, they may substantially impact first-order
	inferential procedures, like Wald tests, resulting in spuriously
	strong or weak conclusions; see \citet[Section~2.1]{chung+etal:2013}
	for an excellent discussion. In contrast to the numerous approaches to
	detect (see, for example, \citealt{kosmidis+schumacher:2021} for the
	\texttt{detectseparation} R \citep{R} package that implements the
	methods in \citealt{konis:2017}) and handle \citep[see, for
	example,][]{kosmidis+firth:2021, heinze+schemper:2002,
		gelman+etal:2008} infinite estimates in logistic regression with
	fixed effects, little methodology or guidance is available on how to
	detect or deal with degenerate estimates in logistic regression with
	mixed effects. For this reason, it is practically desirable to have
	access to methods that are guaranteed to return estimates away from
	the boundary of the parameter space, while preserving the key
	properties that the maximum likelihood estimator has.
	
	We introduce a maximum softly-penalized (approximate) likelihood
	(MSPL) procedure for mixed effects logistic regression that returns
	estimators that are guaranteed to take values in the interior of the
	parameter space, and are also consistent, asymptotically normally
	distributed, and Cram\'{e}r-Rao efficient under assumptions that are
	typically employed for establishing consistency, and asymptotic
	normality of ML estimators. The composite penalty we propose consists
	of appropriately scaled versions of Jeffreys' invariant prior for the
	model with no random effects to ensure the finiteness of the fixed
	effects estimates, and of compositions of the negative Huber loss
	functions to prevent variance components estimates on the boundary of
	the parameter space. We show that the MSPL estimates are guaranteed to
	be in the interior of the parameter space, and scale the penalty
	appropriately to guarantee that i) penalization is soft enough for the
	MSPL estimator to have the same optimal asymptotic properties expected
	by the ML estimator and ii) that the fixed effects estimates are
	equivariant under linear transformations of the model parameters, such
	as contrasts, in the sense that the MSPL estimates of linear
	transformations of the fixed effects parameters are the linear
	transformations of the MSPL estimates. Other prominent penalization
	procedures, for which open-source software implementations exist (for
	example, the bglmer routine of the \texttt{blme} R package; see
	\citep{chung+etal:2013,chung+etal:2015} for details) do not
	necessarily have these properties. Maximum softly-penalized likelihood
	is compared to prominent competing approaches through two real-data
	examples and comprehensive simulation studies that illustrate its
	superior finite-sample performance. Although the developments here are
	for logistic regression with mixed effects, they provide a blueprint
	for the construction of penalties and estimators of the fixed effects
	and/or the variance components with values in the interior of the
	parameter space for any GLMM and, more generally, for M-estimation
	settings where boundary estimates occur.
	
	The remainder of the paper is organized as follows. Section
	\ref{sec:bern_GLMMs} defines the mixed effects logistic regression
	model and Section~\ref{sec:culcita_dat} gives a motivating real-data
	example of degenerate ML estimates in mixed effect logistic
	regression. Section~\ref{sec:composite_penalty} introduces the
	proposed composite penalty, which gives non-boundary MSPL estimates
	(Section~\ref{sec:non_boundary}), is equivariant under linear
	transformations of fixed effects (Section~\ref{sec:invariance}) and
	achieves ML asymptotics (Section~\ref{sec:asymptotics}). Section
	\ref{sec:ci} demonstrates the performance of the MSPL on another
	real-data example for mixed effects logistic regression with bivariate
	random effect structure and presents the results of a simulation study
	based on the data set and Section~\ref{sec:sum} provides concluding
	remarks. Proofs are provided in the Appendix to this paper. Further
	material related to the examples and the simulations is given in the
	accompanying Supplementary Material document, along with additional
	simulation studies that illustrate the relative performance of MSPL to
	alternative methods in artificial mixed effects logistic regression
	settings with extreme fixed effects and variance components.
	
	\section{Mixed effects logistic regression}
	\label{sec:bern_GLMMs}
	
	Suppose that response vectors $\by_1, \ldots, \by_k$ are observed,
	where $\by_i = (y_{i1}, \ldots, y_{in_i})^\top \in \{0, 1\}^{n_i}$,
	possibly along with covariate matrices $\bV_1, \ldots, \bV_k$,
	respectively, where $\bV_i$ is a $n_i \times s$ matrix. A logistic
	regression model with mixed effects assumes that
	$\by_1, \ldots, \by_k$ are realizations of independent random vectors
	$\bY_1, \ldots, \bY_k$, where the entries $Y_{i1}, \ldots, Y_{in_i}$
	are independent Bernoulli random variables conditionally on a vector
	of random effects $\bu_i$ $(i = 1, \ldots, k)$. The vectors
	$\bu_1, \ldots, \bu_k$ are assumed to be independent draws from a
	multivariate normal distribution. The conditional mean of each
	Bernoulli random variable $Y_{ij}$ is associated with a linear
	predictor $\eta_{ij}$, which is a linear combination of covariates
	with fixed and random effects, through the logit link
	function. Specifically,
	\begin{align}
	\label{eq:bern_cluster}
	Y_{ij} \mid \bb{u}_i & \sim \text{Bernoulli}(\mu_{ij}) \quad \text{with} \quad
	\log\frac{\mu_{ij}}{1 - \mu_{ij}} = \eta_{ij} = \bx_{ij}^\top \bbeta + \bz_{ij}^\top \bu_i\\ \notag
	\bu_i & \sim \text{N}(\b0_q, \bb{\Sigma})  \quad (i = 1, \ldots, k; j = 1, \ldots, n_i)\,,
	\end{align}
	where $\mu_{ij} = P(Y_{ij} = 1 \mid \bu_i, \bx_{ij}, \bz_{ij})$.
	The vector $\bx_{ij}$ is the $j$th row of the $n_i \times p$ model
	matrix $\bX_{i}$ associated with the $p$-vector of fixed effects
	$\bbeta \in \Re^p$, and $\bz_{ij}$ is the $j$th row of the
	$n_i \times q$ model matrix $\bZ_{i}$ associated with the $q$-vector
	of random effects $\bu_i$. The model matrices $\bX_i$ and $\bZ_i$ are
	formed from subsets of columns of $\bV_i$, and the matrices $\bX$ and
	$\bV$ with row blocks $\bX_1, \ldots, \bX_n$ and
	$\bV_1, \ldots, \bV_n$ are assumed to be full rank. The
	variance-covariance matrix $\bSigma$ collects the variance components
	and is assumed to be symmetric and positive definite. The marginal
	likelihood about $\bb{\beta}$ and $\bb{\Sigma}$ for
	model~(\ref{eq:bern_cluster}) is
	\begin{equation}
	\label{eq:bern_likl}
	L(\bbeta,\bSigma) = (2\pi)^{-kq/2} \det(\bSigma)^{-k/2} \prod_{i=1}^{k}\int_{\Re^q}\prod_{j=1}^{n_i} \mu_{ij}^{y_{ij}}(1-\mu_{ij})^{1-y_{ij}} \exp\left\{-\frac{\bu_i^\top\bSigma^{-1}\bu_i}{2} \right\} d\bu_i\, .
	\end{equation}
	Formally, the ML estimator is the maximizer of~(\ref{eq:bern_likl})
	with respect to $\bbeta$ and $\bSigma$. However, (\ref{eq:bern_likl})
	involves intractable integrals, which are typically approximated
	before maximization. For example, the popular glmer routine of the R
	package \texttt{lme4} \citep{bates+etal:2015} uses adaptive
	Gauss-Hermite quadrature for one-dimensional random effects and
	Laplace approximation for higher-dimensional random effects. A
	detailed account of those approximation methods can be found in
	\citet{pinheiro+bates:1995}; see also \citet{liu+pierce:1994} for
	adaptive Gauss-Hermite quadrature rules.
	
	\section{Motivating example}
	\label{sec:culcita_dat}
	
	The following section provides a real-data working example, which is a reduced version of the data in \citet{mckeon+etal:2012} as provided in the worked examples of
	\citet{bolker:2015}, that motivates our developments in this paper.
	
	The data, which is given in Table~\ref{tab:culcita},
	comes from a randomized complete block design involving coral-eating sea stars novaeguineae (hereafter Culcita) attacking coral that harbor
	differing combinations of protective symbionts. There are four treatments, namely no symbionts, crabs only, shrimp only, both crabs and shrimp, and ten temporal blocks with two replications per block and treatment, which gives a total of 80 observations on whether predation was present
	(recorded as one) or not (recorded as zero). By mere inspection of the
	responses in Table~\ref{tab:culcita}, we note that predation becomes
	more prevalent with increasing block number and that predation gets
	suppressed when either crabs or shrimp are present, and more so when
	both symbionts are present. The only observation that deviates from
	this general trend is the observation in block 10 with no predation
	and no symbionts. 
	
	\begin{table}[t]
		\caption{Culcita data \citep{mckeon+etal:2012} from the worked
			examples of \citet{bolker:2015}. The data is available at
			\url{https://bbolker.github.io/mixedmodels-misc/ecostats\_chap.html}}
		\label{tab:culcita}
		\begin{center}
			\begin{tabular}{lllllllllll}
				\toprule
				& \multicolumn{10}{c}{Block} \\ \cmidrule{2-11}
				\multicolumn{1}{c}{Treatment} & \multicolumn{1}{c}{1} & \multicolumn{1}{c}{2} & \multicolumn{1}{c}{3} & \multicolumn{1}{c}{4} & \multicolumn{1}{c}{5} & \multicolumn{1}{c}{6} & \multicolumn{1}{c}{7} & \multicolumn{1}{c}{8} & \multicolumn{1}{c}{9} & \multicolumn{1}{c}{10} \\
				\midrule
				none & 0,1 & 1,1 & 1,1 & 1,1 & 1,1 & 1,1 & 1,1 & 1,1 & 1,1 & 1,0 \\
				crabs & 0,0 & 0,0 & 0,0 & 0,0 & 1,1 & 1,1 & 1,1 & 1,1 & 1,1 & 1,1 \\
				shrimp & 0,0 & 0,0 & 0,0 & 0,0 & 0,1 & 1,1 & 1,1 & 1,1 & 1,1 & 1,1 \\
				both & 0,0 & 0,0 & 0,0 & 0,0 & 0,0 & 0,1 & 1,1 & 1,1 & 1,1 & 1,1 \\
				\bottomrule
			\end{tabular}
		\end{center}
	\end{table}
	
	A logistic regression model with one random intercept per block can be
	used here to associate predation with treatment effects while
	accounting for heterogeneity between blocks, i.e.
	\begin{align}
	\label{eq:logistic_normal}
	Y_{ij} \mid u_i & \sim \text{Bernoulli}(\mu_{ij}) \quad  \text{with} \quad
	\log{\frac{\mu_{ij}}{1 - \mu_{ij}}} =  \eta_{ij} = \beta_0 + u_i + \beta_{a(j)} \\ \notag	\label{eq:logistic_normal2}
	u_i & \sim \text{N}(0, \sigma^2) \quad (i = 1, \ldots, 10; j = 1, \ldots, 8) \,,
	\end{align}
	where $a(j) = \lceil j/2 \rceil$ is the ceiling of $j/2$. In the above
	expressions, $(Y_{i1}$, $Y_{i2})^\top$, $(Y_{i3}, Y_{i4})^\top$,
	$(Y_{i5}, Y_{i6})^\top$, $(Y_{i7}, Y_{i8})^\top$ correspond to the two
	responses for each of ``none'', ``crabs'', ``shrimp'', and ``both'',
	respectively. We set $\beta_1 = 0$ for identifiability purposes,
	effectively using ``none'' as a reference category. The logarithm of
	the likelihood~(\ref{eq:bern_likl}) about the parameters
	$\bbeta = (\beta_0, \beta_2, \beta_3, \beta_4)^\top$ and
	$\psi = \log\sigma$ of model~(\ref{eq:logistic_normal}) is
	approximated by an adaptive quadrature rule as implemented in glmer
	with $Q = 100$ quadrature points, which is the maximum the current
	glmer implementation allows. All parameter estimates of
	model~(\ref{eq:logistic_normal}) reported in the current example are
	computed after removal of the atypical observation with zero predation
	in block 10 when there are no symbionts. This is also done in
	\citet[Section~13.5.6]{bolker:2015} and the corresponding worked
	examples, which are available at
	\url{https://bbolker.github.io/mixedmodels-misc/ecostats\_chap.html}. Estimates
	based on all data points are provided in Table~\ref{tab:culcita_supp}
	of the Supplementary Material document.
	
	The ML estimates of $\bbeta$ and $\psi$ in Table~\ref{tab:culcita_inf}
	are computed using the numerical optimization procedures BFGS and CG
	(ML[BFGS] and ML[CG], respectively), as these are provided by the
	\texttt{optimx} R package \citep[see][Section~3 for
	details]{nash:2014}, with the same starting values. The ML[BFGS] and
	ML[CG] estimates are different and notably extreme on the logistic
	scale. This is due to the two optimization procedures stopping early
	at different points in the parameter space after having prematurely
	declared convergence. The large estimated standard errors are
	indicative of an almost flat approximate log-likelihood around the
	estimates. In this case, the ML estimates of the fixed effects
	$\beta_0, \beta_2, \beta_3, \beta_4$ are in reality infinite in
	absolute value, which has also been observed in
	\citet[Section~13.5.6]{bolker:2015}.
	
	\begin{table}[t]
		\caption{ML, bglmer, and MSPL parameter estimates of
			\eqref{eq:logistic_normal} from the Culcita data in
			Table~\ref{tab:culcita} after removing the observation with zero
			predation in block 10 when there are no symbionts. The likelihood is approximated using a 100-point adaptive Gauss-Hermite quadrature rule. Parameter estimates are reported for the model with
			$\eta_{ij}$ as in~\eqref{eq:logistic_normal}, and with
			$\eta_{ij} = \gamma_0 + u_i + \gamma_j$ with $\gamma_4 =
			0$. Estimated standard errors (in parentheses) are based on the
			negative Hessian of the approximate log-likelihood}
		\label{tab:culcita_inf}
		\begin{center}
			\begin{tabular}{lD{.}{.}{3}D{.}{.}{3}D{.}{.}{3}D{.}{.}{3}D{.}{.}{3}}
				\toprule
				&
				\multicolumn{1}{c}{ML[BFGS]} & 
				\multicolumn{1}{c}{ML[CG]} &
				\multicolumn{1}{c}{bglmer[t]} &
				\multicolumn{1}{c}{bglmer[n]} & \multicolumn{1}{c}{MSPL}
				\\
				\midrule
				\multicolumn{6}{c}{reference category: ``none''} \\
				\midrule
				$\beta_0$ & 15.88 & 15.37 & 6.39 & 4.90 & 8.05\\
				& (10.14) & (9.50) & (2.60) & (2.08) & (3.21)\\
				$\beta_2$ & -12.93 & -12.46 & -4.02 & -2.84 & -6.90\\
				& (9.15) & (8.53) & (1.59) & (1.27) & (3.00)\\
				$\beta_3$ & -14.81 & -14.30 & -4.81 & -3.44 & -7.87\\
				& (9.89) & (9.24) & (1.73) & (1.35) & (3.26)\\
				$\beta_4$ & -17.71 & -17.15 & -6.47 & -4.73 & -9.64\\
				& (10.70) & (10.02) & (2.05) & (1.57) & (3.61)\\
				$\log \sigma$ & 2.31 & 2.28 & 1.72 & 1.54 & 1.72\\
				& (0.64) & (0.62) & (0.44) & (0.43) & (0.44)\\
				\midrule
				\multicolumn{6}{c}{reference category: ``both''} \\
				\midrule
				$\gamma_0$ & -1.82 & -1.74 & 0.37 & 0.57 & -1.59\\
				& (3.92) & (3.77) & (2.24) & (2.07) & (2.28)\\
				$\gamma_1$ & 17.74 & 17.09 & 6.70 & 5.75 & 9.63\\
				& (10.75) & (10.03) & (2.19) & (1.88) & (3.61)\\
				$\gamma_2$ & 4.78 & 4.65 & 1.63 & 1.26 & 2.74\\
				& (3.08) & (2.98) & (1.43) & (1.32) & (1.79)\\
				$\gamma_3$ & 2.89 & 2.83 & 0.83 & 0.56 & 1.77\\
				& (2.27) & (2.22) & (1.35) & (1.28) & (1.55)\\
				$\log \sigma$ & 2.31 & 2.28 & 1.74 & 1.66 & 1.72\\
				& (0.64) & (0.62) & (0.44) & (0.44) & (0.44)\\
				\bottomrule
			\end{tabular}
		\end{center}
	\end{table}

	Parameter estimates are also obtained using the bglmer routine of
	the \texttt{blme} R package \citep{chung+etal:2013} that has been
	developed to ensure that parameter estimates from GLMMs are away from
	the boundary of the parameter space. The estimates shown in
	Table~\ref{tab:culcita_inf} are obtained using a penalty for $\sigma$
	inspired by a gamma prior (default in bglmer; see
	\citealt{chung+etal:2013} for details) and two of the default prior
	specifications for the fixed effects: i) independent normal priors
	(bglmer[n]), and ii) independent t priors (bglmer[t]), as these are
	implemented in \texttt{blme}; see \texttt{bmerDist-class} in the help
	pages of \texttt{blme} for details. We also show the estimates
	obtained using the MSPL estimation method that we propose in the
	current work.
	
	The maximum penalized likelihood estimates from bglmer and the
	corresponding estimated standard errors appear to be
	finite. Nevertheless, the use of the default priors directly breaks
	parameterization equivariance under contrasts, which ML estimates
	enjoy. For example, Table~\ref{tab:culcita_inf} also shows the
	estimates of model~(\ref{eq:logistic_normal}) with
	$\eta_{ij} = \gamma_0 + u_i + \gamma_{j}$, where $\gamma_{4} = 0$,
	i.e.~setting ``both'' as a reference category. Hence, the identities
	$\gamma_0 = \beta_0 + \beta_4$, $\gamma_1 = -\beta_4$,
	$\gamma_2 = \beta_2 - \beta_4$, $\gamma_3 = \beta_3 - \beta_4$ hold,
	and it is natural to expect those identities from the estimates of
	$\bbeta$ and $\bgamma$. As is evident from
	Table~\ref{tab:culcita_inf}, the bglmer estimates with either normal
	or t priors can deviate substantially from those identities. For
	example, the bglmer estimate of $\gamma_1$ based on normal priors is
	$5.75$ while that for $\beta_4$ is $-4.73$, and the estimate of
	$\log\sigma$ is $1.54$ in the $\bbeta$ parameterization and $1.66$ in
	the $\bgamma$ parameterization. Furthermore, different contrasts give
	varying amounts of deviations from these identities. On the other
	hand, the approximate likelihood is invariant under monotone parameter
	transformations. As a result, the corresponding identities hold
	exactly for the ML estimates with the deviations observed in
	Table~\ref{tab:culcita_inf} being due to early stopping of the
	optimization routines. The bglmer estimates are typically closer to
	zero in absolute value than the ML estimates because the normal and t
	priors are all centered at zero. Note that the estimates using normal
	priors tend to shrink more towards zero than those using t priors
	because the latter have heavier tails.

	\begin{figure}[t]
		\begin{center}
			\includegraphics[width = 0.8\textwidth]{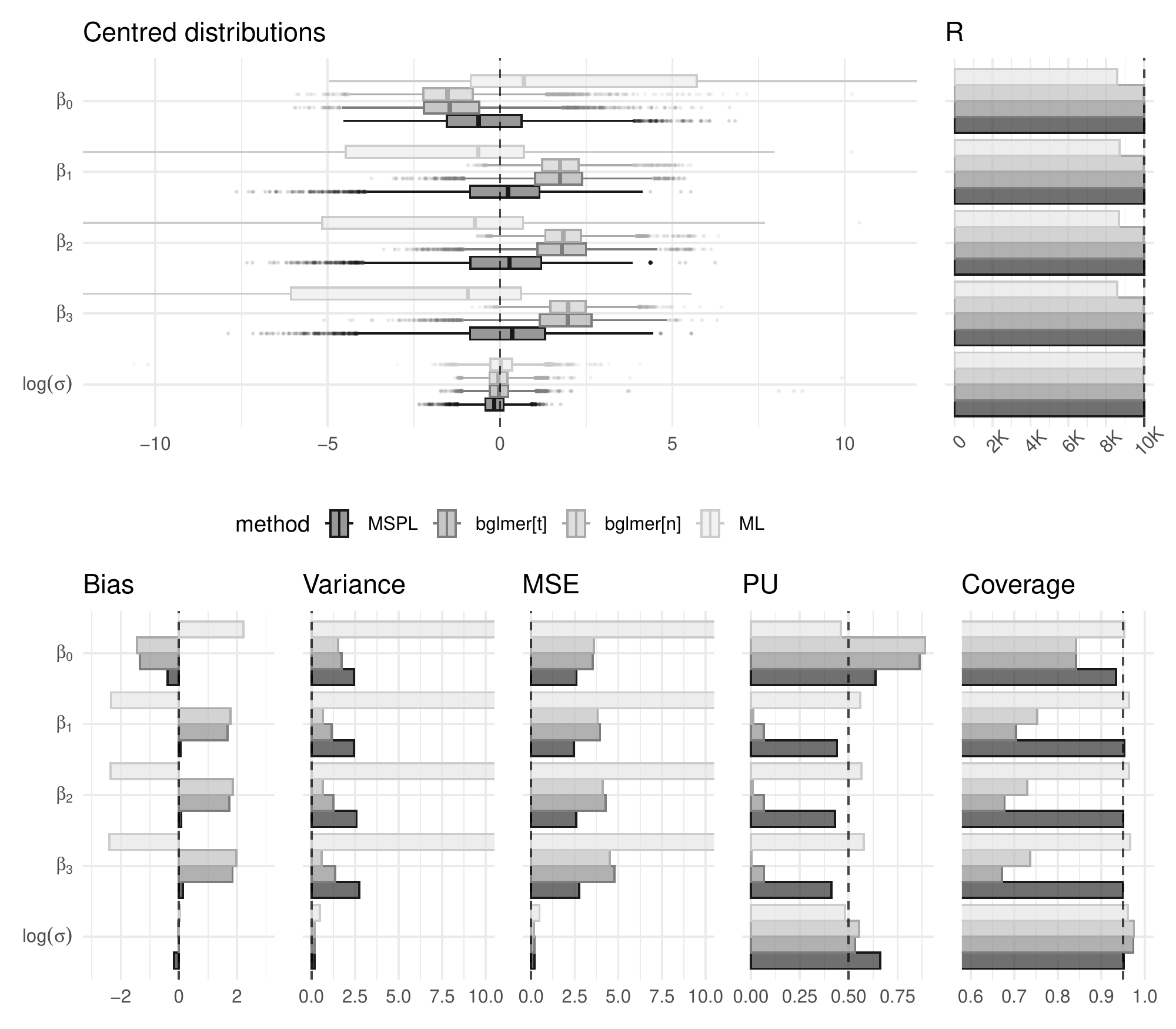}
		\end{center}
		\caption{Simulation results based on $10\ 000$ samples from
			\eqref{eq:logistic_normal} at the ML estimates in
			Table~\ref{tab:culcita_supp} in the Supplementary Material
			document. Estimates are obtained using ML, bglmer, and MSPL, with
			a 100-point adaptive Gauss-Hermite quadrature approximation to the
			likelihood. Parameter estimates on the boundary are discarded for
			the calculation of the simulation summaries. The number of
			estimates used for the calculation of summaries is given in the
			top right panel (R). The top left panel shows the centered
			sampling distribution of the estimators for MSPL, bglmer and
			ML. The bottom panels give simulation-based estimates of the bias,
			variance, mean squared error (MSE), probability of underestimation
			(PU), and coverage based on 95\% Wald-confidence intervals}
		\label{fig:culcita_simu0}
	\end{figure}

	In order to assess the impact of shrinkage on the frequentist
	properties of the estimators, we simulate $10\ 000$ independent
	samples of responses for the randomized complete block design in
	Table~\ref{tab:culcita}, at the ML estimates in the $\bbeta$
	parameterization, when all data points are used (see
	Table~\ref{tab:culcita_supp} of the Supplementary Material
	document). For each sample, we compute the ML and MSPL estimates, as
	well as the bglmer estimates based on normal and t priors.
	Figure~\ref{fig:culcita_simu0} shows boxplots for the sampling
	distributions of the estimators, centered at the true values, the
	estimated finite-sample bias, variance, mean squared error, and
	probability of underestimation for each estimator, along with the
	estimated coverage of 95\% Wald confidence intervals based on the
	estimates and estimated standard errors from the negative Hessian of
	the approximate log-likelihood at the estimates. The plotting range
	for the support of the distributions has been restricted to
	$(-11, 11)$, which does not contain all ML estimates in the simulation
	study but contains all estimates for the other methods. Note that
	apart from the estimated probability of underestimation, estimates for
	the other summaries are not well-defined for ML, because the
	probability of boundary estimates is positive. In fact, there were
	issues with at least one of the ML estimates for $9.25\%$ of the
	simulated samples. These issues are either due to convergence failures
	or because the estimates or estimated standard errors have been found
	to be atypically large in absolute value. The displayed summaries for
	ML are computed based only on estimates which have not been found to
	be problematic. Clearly, the amount of shrinkage induced by the normal
	and t priors is excessive. Although the resulting estimators have
	small finite-sample variance (with the one based on normal priors
	having the smallest), they have excessive finite-sample bias, which is
	often at the order of the standard deviation.  The combination of
	small variance and large bias results in large mean squared errors,
	and the sampling distributions to be located far from the respective
	true values, impacting first-order inferences; Wald-type confidence
	intervals about the fixed effects are found to systematically
	undercover the true parameter value. Finally, neither bglmer[n] nor
	bglmer[t] appear to prevent extreme positive variance estimates.

	As is apparent from Table~\ref{tab:culcita_inf}, the MSPL estimates
	are equivariant under contrasts. The identities between the
	$\bbeta$ and $\bgamma$ parameterization of the model hold
	exactly for the proposed MSPL estimates, where the small observed 
	deviations are attributed to rounding and the stopping criteria used
	for the numerical optimization of the penalized
	log-likelihood. Furthermore, we see from Figure~\ref{fig:culcita_simu0} 
	that the penalty we propose not only ensures that estimates are
	away from the boundary of the parameter space, but its soft nature
	leads to estimators that possess the good frequentist properties that
	would be expected by the ML estimator had it not taken boundary
	values.
	
	\section{Composite penalty}\label{sec:composite_penalty}
	We define a penalty for the log-likelihood or an approximation thereof
	for mixed effects logistic regression that returns MPL estimators that
	are always in the interior of the parameter space and are equivariant
	under scaled contrasts of the fixed effects. The penalty is
	appropriately scaled to be soft enough to return MPL estimators that
	are consistent and asymptotically normally distributed. For this
	reason, the resulting MPL estimators are termed maximum
	softly-penalized likelihood (MSPL) estimators.
	
	Let $\btheta = (\bbeta^\top, \bpsi^\top)^\top$ and
	$\ell(\btheta) = \log L(\bbeta, s(\bpsi))$ with $s(\bpsi) = \bSigma$,
	where $L(\bbeta, s(\bpsi))$ is~\eqref{eq:bern_likl}. For clarity of
	presentation, we shall write $\ell(\btheta)$ to denote both the
	log-likelihood or an appropriate approximation of the log-likelihood
	that is bounded from above. Sufficient conditions for consistency and
	asymptotic normality of the MSPL estimator using the exact likelihood
	and an approximation thereof are provided in Section~\ref{sec:softpen}
	of the Appendix. The parameter vector $\bpsi$ is defined as
	$\bpsi = (\log l_{11}, \ldots, \log l_{qq}, l_{21}, \ldots, l_{q1},
	l_{32}, \ldots, l_{q2}, \ldots, l_{qq-1})^\top$, where $l_{ij}$
	$(i > j)$ is the $(i,j)$th element of the lower-triangular Cholesky
	factor $\bL$ of $\bSigma$, i.e.  $\bSigma = \bL\bL^\top$. Consider the
	estimator
	\[
	\tilde{\btheta} = \arg\max_{\btheta \in \ \Theta} \left\{\ell(\btheta) + c_1 P_{(f)}(\bbeta) + c_2 P_{(v)}(\bpsi) \right\}\, ,
	\]
	where $c_{1} > 0$, $c_{2} > 0$, and $P_{(f)}(\bbeta)$ and
	$P_{(v)}(\bpsi)$ are unscaled penalty functions for the fixed effects
	and variance components, respectively. 
	
	For the unscaled fixed effects penalty, we use the logarithm of
	Jeffreys' prior for the corresponding GLM, that is
	\begin{equation}
	\label{eq:jeffreys_pen}
	P_{(f)}(\bbeta) = \frac{1}{2}\log \det\left(\sum_{i = 1}^k \bX^\top_i \bW_i \bX_i\right)\,,
	\end{equation}
	where $\bX_i$ collects the covariates for the fixed effects in
	model~(\ref{eq:bern_cluster}), and $\bW_i$ is a diagonal matrix with
	$j$th diagonal element $\mu_{ij}^{(f)} (1 - \mu_{ij}^{(f)})$ with
	$\mu_{ij}^{(f)} = \exp(\eta_{ij}^{(f)}) / \{1 +
	\exp(\eta_{ij}^{(f)})\}$ and $\eta_{ij}^{(f)} = \bx_{ij}^\top
	\bbeta$. For the variance components penalty, we use a composition of
	negative Huber loss functions on the components of $\bpsi$. In
	particular,
	\begin{equation}
	\label{eq:huber_pen}
	P_{(v)}(\bpsi) = \sum_{i = 1}^q D(\log l_{ii}) + \sum_{i > j} D(l_{ij}) \, ,
	\end{equation}
	where
	\[
	D(x) = \begin{cases}
	-\frac{1}{2} x^2, & \text{if } |x|\leq 1 \\ 
	- |x| + \frac{1}{2}, & \text{otherwise}           
	\end{cases} \,.
	\]
	
	\section{Non-boundary MPL estimates}
	\label{sec:non_boundary}
	
	Denote by $\partial \Theta$ the boundary of $\Theta$ and let
	$\btheta(r)$, $r \in \Re$, be a path in the parameter space such that
	$\lim_{r \to \infty}\btheta(r) \in \partial \Theta$. A common approach
	to resolving issues with ML estimates being in $\partial \Theta$, like
	those encountered in the example of Section~\ref{sec:culcita_dat}, is
	to introduce an additive penalty to the (approximate) log-likelihood
	that satisfies $\lim_{r \to \infty} P(\btheta(r)) = -\infty$.  Hence,
	if $\ell(\btheta)$ is bounded from above and there is at least one
	point $\btheta \in \Theta$ such that
	$\ell(\btheta) + P(\btheta) > -\infty$, then $\tilde\btheta$ is in the
	interior of $\Theta$.
	
	\citet[Theorem 1]{kosmidis+firth:2021} show that if the matrix $\bX$
	with row blocks $\bX_1, \ldots, \bX_n$ is full rank, then the limit
	of~(\ref{eq:jeffreys_pen}) is $-\infty$ as $\bbeta$ diverges to any point
	with at least one infinite component. This result holds for a range of
	link functions including the commonly-used logit, probit,
	complementary log-log, log-log, and the cauchit link
	\citep[see][Section~3.1, for details]{kosmidis+firth:2021}. Now,
	noting that~\eqref{eq:bern_likl} is always bounded from above by one
	as a probability mass function, the penalized log-likelihood
	$\ell(\btheta) + P(\btheta)$ diverges to $-\infty$ as $\bbeta$
	diverges, for any value of $\bpsi$. Hence, the MPL estimates for the
	fixed effects always have finite components as long as there is at
	least one point in $\Theta$ such that $\ell(\btheta)$ is not $-\infty$.
	
	The penalty~\eqref{eq:huber_pen} on the variance
	components takes value $-\infty$ whenever at least one component of
	$\bpsi$ diverges.  Hence, by parallel arguments to those in the
	previous paragraph, the penalized log-likelihood
	$\ell(\btheta) + P(\btheta)$ diverges to $-\infty$ as any
	component of $\bpsi$ diverges, for any value of $\bbeta$. Hence, the
	MPL estimates for $\bpsi$ have finite components and the value of
	$\tilde\bSigma = s(\tilde{\bpsi})$ is guaranteed to be non-degenerate
	in the sense that it is positive definite with finite entries, implying
	correlations away from one in absolute value (see
	Theorem~\ref{thm:nondeg}), as long as there is at least one point in
	$\Theta$ such that $\ell(\btheta)$ is not $-\infty$.
	
	The condition on the boundedness of~\eqref{eq:bern_likl} from above is
	just one sufficient condition for the finiteness of the MPL estimates,
	which is also satisfied by a vanilla (non-adaptive) Gauss-Hermite
	quadrature or simulation-based approximations of the likelihood
	\citep[see, for example,][]{mcculloch:1997}. A weaker sufficient
	condition is that the penalized objective diverges to $-\infty$ as
	$\btheta(r)$ diverges to $\partial \Theta$, or, in other words, that
	the penalty dominates the likelihood in absolute value for any
	divergent path. From the numerous numerical experiments we carried
	out, we encountered no evidence that this weaker condition does not
	hold for the adaptive quadrature and Laplace approximations to the
	log-likelihood that the glmer routine of the R package \texttt{lme4}
	employs.
	
	The penalties arising from the independent normal and independent t
	prior structures implemented in \texttt{blme} are such that
	$\lim_{r \to \infty} P(\btheta(r)) = -\infty$, whenever $\btheta(r)$
	diverges to the boundary of the parameter space for the fixed
	effects. As a result, the bglmer[n] and bglmer[t] estimates for the
	fixed effects are always finite, as also illustrated in
	Table~\ref{tab:culcita_inf}. Nevertheless, the default gamma-prior
	like penalty used in bglmer for the variance component $\sigma$ is
	$1.5 \log\sigma$, which, while it ensures that the estimate of
	$\log \sigma$ is not minus infinity, does not guard against estimates
	that are $+\infty$. This is apparent in
	Figure~\ref{fig:culcita_simu0}, where several extreme positive
	bglmer[n] and bglmer[t] estimates are observed for $\log\sigma$.
	
	\section{Equivariance under linear transformations of fixed effects}\label{sec:invariance}
	The ML estimates are known to be equivariant under transformations of
	the model parameters \citep[see, for example][]{zehna:1966}.  A
	particularly useful class of transformations in mixed effects logistic
	regression, and more generally GLMMs with categorical covariates, is
	the collection of scaled linear transformations $\bbeta' = \bC \bbeta$
	of the fixed effects for known, invertible, real matrices $\bC$.
	
	Such invariance properties of the ML estimates guarantee that one can
	obtain ML estimates and corresponding estimated standard errors for
	arbitrary sets of scaled parameter contrasts of the fixed effects,
	when estimates for one of those sets of contrasts are available and
	with no need to re-estimate the model. Such equivariance properties
	eliminate any estimation and inferential ambiguity when two
	independent researchers analyze the same data set using the same model
	but with different contrasts for the fixed effects, for example, due
	to software defaults.
	
	Following the argument in \citet{zehna:1966}, the
	condition required for achieving equivariance of MPL estimators is that
	the penalty for the fixed effects parameters behaves like the
	log-likelihood under linear transformations; that is
	\begin{equation}
	\label{eq:invariance}
	P_{(f)}(\bC\bbeta) = P_{(f)}(\bbeta) + a,
	\end{equation}
	where $a \in \Re$ is a scalar that does not depend on $\bbeta$.
	
	Let $\eta_{ij} = \bx_{ij}^\top \bC^{-1} \bgamma + \bz_{ij}^\top \bu_i$
	in~\eqref{eq:bern_cluster} for a known real matrix $\bC$. Then,
	$\bgamma = \bC\bbeta$, and the penalty for the fixed effects in the
	$\bgamma$ parameterization is
	$P_{(f)}(\bgamma) = P_{(f)}(\bbeta) - \log \det (\bC)$, which is of
	the form of~\eqref{eq:invariance}. In contrast, the penalties arising
	from the normal and t prior structures used to compute the bglmer[n]
	and bglmer[t] fixed effect estimates in Table~\ref{tab:culcita_inf} do
	not satisfy~\eqref{eq:invariance}. Hence, the bglmer[n] and bglmer[t]
	estimates are not equivariant under linear transformations of the
	parameters.
	
	\section{Consistency and asymptotic normality of the MSPL estimator}
	\label{sec:asymptotics}
	To mitigate any distortion of the estimates by the penalization of the
	log-likelihood, and preserve ML asymptotics, we choose the scaling
	factors $c_1,c_2$ to be ``soft'' enough to control
	$\vnorm{\nabla_{\btheta} P(\boldsymbol{\theta})}$ in terms of the rate
	of information accumulation
	$\bb{r}_n = (r_{n1},\ldots,r_{nd})^\top \in \Re^d$ with
	$d = p + q(q + 1)/2$. The components of the rate of information
	accumulation are defined to diverge to $+\infty$ and satisfy
	$R_n^{-1/2}J(\btheta)R_n^{-1/2} \overset{p}{\to} I(\btheta)$ as
	$n \to \infty$, where $R_n$ is a diagonal matrix with the elements of
	$\bb{r}_n$ on its main diagonal,
	$J(\boldsymbol{\theta}) = -\nabla_{\btheta} \nabla_{\btheta}^\top
	\ell(\boldsymbol{\theta})$ is the observed information matrix, and
	$I(\boldsymbol{\theta})$ is a $\mathcal{O}(1)$ matrix. In this way, we
	allow for scenarios where the rate of information accumulation varies
	across the components of the parameter vector.   
	
	According to Theorem~\ref{thm:jeffrey_deriv_bound} in the Appendix,
	the gradient of the logarithm of the Jeffreys' prior
	in~\eqref{eq:jeffreys_pen} can be bounded as
	$\vnorm{\nabla_{\bbeta} P_{(f)}(\bbeta)} \le p^{3/2} \max_{s,t}
	|x_{st}|/2$, where $x_{st}$ is the $t$th element in the $s$th row of
	$\bX$ as defined in Section~\ref{sec:non_boundary} and
	$\vnorm{\cdot}$ is the Euclidean norm. Furthermore,
	$\vnorm{\nabla_{\bpsi} P_{(v)}(\bpsi)} \le \sqrt{q(q+1)/2}$ because
	$|d D(x) / dx| \le 1$. Hence, an application of the triangle
	inequality gives that
	$\vnorm{\nabla_{\btheta} P(\btheta)} \le c_1 p^{3/2} \max_{s, t}
	|x_{st}| / 2 + c_2 \sqrt{q(q+1) / 2}$. For the scaling factors $c_1$
	and $c_2$, we propose using $c_1 = c_2 = c$ to be the square root of
	the average of the approximate variances of $\hat\eta_{ij}^{(f)}$ at
	$\bbeta = \b0_p$. A delta method argument gives that
	$c = 2 \sqrt{p/n}$. Therefore,
	\[
	\vnorm{\nabla_{\btheta} P(\btheta)} \le \frac{p^{2}}{\sqrt{n}} \max_{s, t} |x_{st}|  +  \sqrt{\frac{2pq(q+1)}{n}} \, .
	\]
	Hence, as long as $\max_{s, t} |x_{st}| = O_p(n^{1/2})$, it holds that
	$ \underset{\btheta\in \Theta}{\sup} \vnorm{R_n^{-1} \nabla_{\btheta}
		P(\btheta)} = o_p(1)$, and the conditions for consistency and
	asymptotic normality of $\bttilde$ in Theorem~\ref{thm:soft_pen_cons}
	and Theorem~\ref{thm:asymp_norm_soft_pen}, respectively, are
	satisfied. The condition on the maximum of the absolute elements of
	the model matrix is not unreasonable in practice. It certainly holds
	true for bounded covariates such as dummy variables, as encountered in
	the real-data examples in the current paper. It is also true for model
	matrices with subgaussian random variables with common variance proxy
	$\sigma^2$, in which case
	$\max_{s,t} |x_{st} | = \Op{\sqrt{2\sigma^2 \log(2np)}}$ \citep[see,
	for example,][Theorem 1.14]{rigollet:2015}.
	
	\section{Conditional inference data} 
	\label{sec:ci}
	To demonstrate the performance of the MSPL mixed effects logistic regression with a multivariate random effects structure, we consider a subset
	of the data analyzed by \citet{singmann+etal:2016}. As discussed on
	CrossValidated
	(\url{https://stats.stackexchange.com/questions/38493}), this data set
	exhibits both infinite fixed effects estimates as well as degenerate
	variance components estimates when a Bernoulli-response GLMM is fitted
	by ML.
	
	The data set, originally collected as a control condition of experiment 3)b) in \citet{singmann+etal:2016} and therein analyzed in a different context, comes from an experiment in which participants worked on a probabilistic conditional inference task. Participants were presented with the conditional inferences modus ponens (MP; ``\textit{If p then q. p. Therefore q.}''), modus tollens (MT; ``\textit{If p then q. Not q. Therefore  not p.}''), affirmation of the consequent (AC; ``\textit{If p then q. q. Therefore p}''), and denial of the antecedent (DA, ``\textit{If p then q. Not p. Therefore not q}''), and asked to estimate the probability that the conclusion (``\textit{Therefore ...}'') follows from the conditional rule (``\textit{If p then q.}'') and the minor premise (``\textit{p}.'', ``\textit{not p}.'', ``\textit{q}.'', ``\textit{not q}.''). The material of the experiment consisted of the following four conditional rules with varying degrees of counterexamples (alternatives, disablers; indicated in parentheses below).
	\begin{enumerate}
		\item If a predator is hungry, then it will search for prey. (few disablers, few alternatives)
		\item If a person drinks a lot of coke, then the person will gain weight. (many disablers, many alternatives)
		\item If a girl has sexual intercourse with her partner, then she will get pregnant. (many disablers, few alternatives)
		\item If a balloon is pricked with a needle, then it will quickly lose air. (few disablers, many alternatives)
	\end{enumerate}
	To illustrate, for MP and conditional rule 1, a participant was asked: ``\textit{If a predator is hungry, then it will search for prey. A predator is hungry. How likely is it that the predator will search for prey?}'' Additionally, participants were asked to estimate the probability of the conditional rule, e.g. ``\textit{How likely is it that if a predator is hungry it will search for prey?}'', and the probability of the minor premises, e.g. ``\textit{How likely is it that a predator is hungry?}''. 
	
	The response variable of this data set is then a binary response indicating whether, given a certain conditional rule, the participants' probabilistic inference is $p$-valid. An inference is deemed $p$-valid if the summed uncertainty of the premises does not exceed the uncertainty of the conclusion, where uncertainty of a statement $x$ is defined as one minus the probability of $x$. For example, for MP, a respondent's inference is $p$-valid if $1-\Pr(\textit{``q''}) \leq 1 - \Pr(\textit{``Ìf p then q''}) + 1-\Pr(\textit{``p''}) $, where $\Pr(x)$ indicates the participant's estimated probability of statement $x$ ($p$-valid inferences are recorded as zero, $p$-invalid inferences as one). Covariates are the categorical variable ``counterexamples'' (``many'', ``few''), that indicates the degree of available counterexamples to a conditional rule, ``type'' (``affirmative'',``denial'') which describes the type of inference (MP and AC are affirmative, MT and DA are denial), and ``$p$-validity'' (``valid'',``invalid''), indicating whether an inference is $p$-valid in standard probability theory where premise and conclusions are seen as events (MP and MT are $p$-valid, while AC and DA are not). For each of the 29 participants, there exist 16 observations corresponding to all possible combinations of inference and conditional rule, giving a total of 464 observations, which are grouped along individuals by the clustering variable ``code''. A mixed effects logistic regression model can be employed to investigate the probabilistic validity of conditional inference given the type of inference and conditional rule as captured by the covariates and all possible interactions thereof. A random intercept and random slope for the variable ``counterexamples'' are introduced to account for response heterogeneity between participants. Hence the model under consideration is given by    
	\begin{align}
	\label{eq:cond_inf_model} 
	Y_{ij} \mid \bb{u}_i & \sim \text{Bernoulli}(\mu_{ij}) \quad \text{with} \quad
	g(\mu_{ij}) = \eta_{ij} = \bx_{ij}^\top \bbeta + \bz_{ij}^\top \bu_i\\
	\bu_i & \sim \text{N}(\b0_2, \bb{\Sigma})  \quad (i = 1, \ldots, 29; j = 1, \ldots, 16)\,,
	\end{align}
	where $\bbeta = (\beta_0,\beta_1,\ldots,\beta_8)$ are the fixed
	effects pertaining to the model matrix of the R model formula
	\texttt{response \raisebox{-0.9ex}{\~{}} type * p.validity *
		counterexamples + (1+counterexamples|code)}.\\ Gauss-Hermite
	quadrature is computationally challenging for multivariate random
	effect structures. For this reason glmer and bglmer do not offer
	it as an option. We approximate the likelihood
	of~\eqref{eq:cond_inf_model} about the parameters $\bbeta$, $\bL$
	using Laplace's method. We estimate the parameters $\bbeta$, $\bL$ by
	ML using the optimization routines CG (ML[CG]) and BFGS (ML[BFGS]) of
	the \texttt{optimx} R package \citep{nash:2014}, bglmer from the
	\texttt{blme} R package \cite{chung+etal:2013} using independent
	normal (bglmer[n]) and t (bglmer[t]) priors for the fixed effects and
	the default Wishart prior for the multivariate variance components. We
	also estimate the parameters using the proposed MSPL estimator with
	the fixed and random effects penalties of
	Section~\ref{sec:composite_penalty}. The estimates are given in
	Table~\ref{tab:cond_inf}, where we denote the entries of $\bL$ by
	$l_{st}$, for $s,t=1,2$.
	\begin{table}[t]
		\caption{ML, bglmer and MSPL estimates for the parameters of
			model~\eqref{eq:cond_inf_model} using a Laplace approximation of
			the log-likelihood. Estimated standard errors (in parentheses) are
			based on the inverse of the negative Hessian of the approximate
			likelihood. The estimated standard errors are not reported when
			the corresponding diagonal elements of the inverse are negative}
		\label{tab:cond_inf}
		\begin{center}
			\begin{tabular}{lD{.}{.}{3}D{.}{.}{3}D{.}{.}{3}D{.}{.}{3}D{.}{.}{3}}
				\toprule
				&
				\multicolumn{1}{c}{ML[BFGS]} & 
				\multicolumn{1}{c}{ML[CG]} &
				\multicolumn{1}{c}{bglmer[t]} &
				\multicolumn{1}{c}{bglmer[n]} & 
				\multicolumn{1}{c}{MSPL} \\
				\midrule
			$\beta_0$ & 16.28 & 7.63 & 10.60 & 10.04 & 6.22 \\ 
			& (2.09) & (3.69) & (1.56) & (1.23) & (2.49) \\ 
			$\beta_2$ & 4.27 & 3.22 & 1.57 & 1.28 & 0.00 \\ 
			& (4.57) & (14.94) & (1.58) & (1.80) & (2.83) \\ 
			$\beta_3$ & -6.75 & -2.08 & 0.12 & 0.35 & -2.17 \\ 
			& (3.02) & (5.19) & (1.59) & (1.40) & (2.60) \\ 
			$\beta_4$ & -14.44 & -5.80 & -8.50 & -7.90 & -4.37 \\ 
			& (2.13) & (3.75) & (1.80) & (1.35) & (2.51) \\ 
			$\beta_5$ & 3.21 & 0.85 & 0.41 & 0.45 & 2.17 \\ 
			& (3.34) & (16.63) &  & (0.35) & (4.07) \\ 
			$\beta_6$ & -4.27 & -3.24 & -1.70 & -1.45 & 0.00 \\ 
			& (4.60) & (14.95) & (1.62) & (1.92) & (2.86) \\ 
			$\beta_7$ & 8.25 & 3.78 & 1.16 & 0.86 & 3.64 \\ 
			& (3.13) & (5.26) & (1.77) & (1.63) & (2.71) \\ 
			$\beta_8$ & -3.94 & -1.81 & -0.89 & -0.86 & -2.87 \\ 
			& (3.49) & (16.67) &  & (1.29) & (4.18) \\ 
			$\log l_{1,1}$ & 2.02 & 0.73 & 3.34 & 3.34 & -0.63 \\ 
			& (0.16) & (1.03) & (0.00) & (0.01) & (2.48) \\ 
			$l_{2,1}$ & -7.67 & -2.25 & -28.35 & -28.13 & -0.60 \\ 
			& (0.97) & (2.29) & (0.12) &  & (1.69) \\ 
			$\log l_{2,2}$ & -5.17 & -2.92 & -0.65 & -0.73 & -1.21 \\ 
			& (83.15) & (8.01) & (0.75) & (0.83) & (1.30) \\ 
				\bottomrule
			\end{tabular}
		\end{center}
	\end{table}
	As in the Culcita example of Section~\ref{sec:culcita_dat}, we
	encounter fixed effects estimates that are extreme on the logistic
	scale for ML[BFGS] and ML[CG]. We note that the strongly negative
	estimates for $l_{22}$ in conjunction with the inflated asymptotic
	standard errors of the ML[BFGS] estimates are highly indicative of
	parameter estimates on the boundary of the parameter space, meaning
	that $l_{22}$ is effectively estimated as zero. The degeneracy of the
	estimated variance components is even more striking for the bglmer
	estimates, which give estimates of $l_{11},l_{21}$ greater than $28$
	in absolute value, which corresponds to estimated variance components
	greater than $800$ in absolute value. This underlines that, as with
	the gamma prior penalty for univariate random effects, the Wishart
	prior penalty, while effective in preventing variance components being
	estimated as zero, cannot guard against infinite estimates for the
	variance components. We finally note that for the MSPL, all parameter
	estimates as well as their estimated standard errors appear to be
	finite. Further, while the variance components penalty guards against
	estimates that are effectively zero, the penalty induced shrinkage
	towards zero is not as strong as with the Wishart prior penalty of the
	bglmer routine.
	
	\begin{figure}[t]
		\begin{center}
			\includegraphics[width=\textwidth]{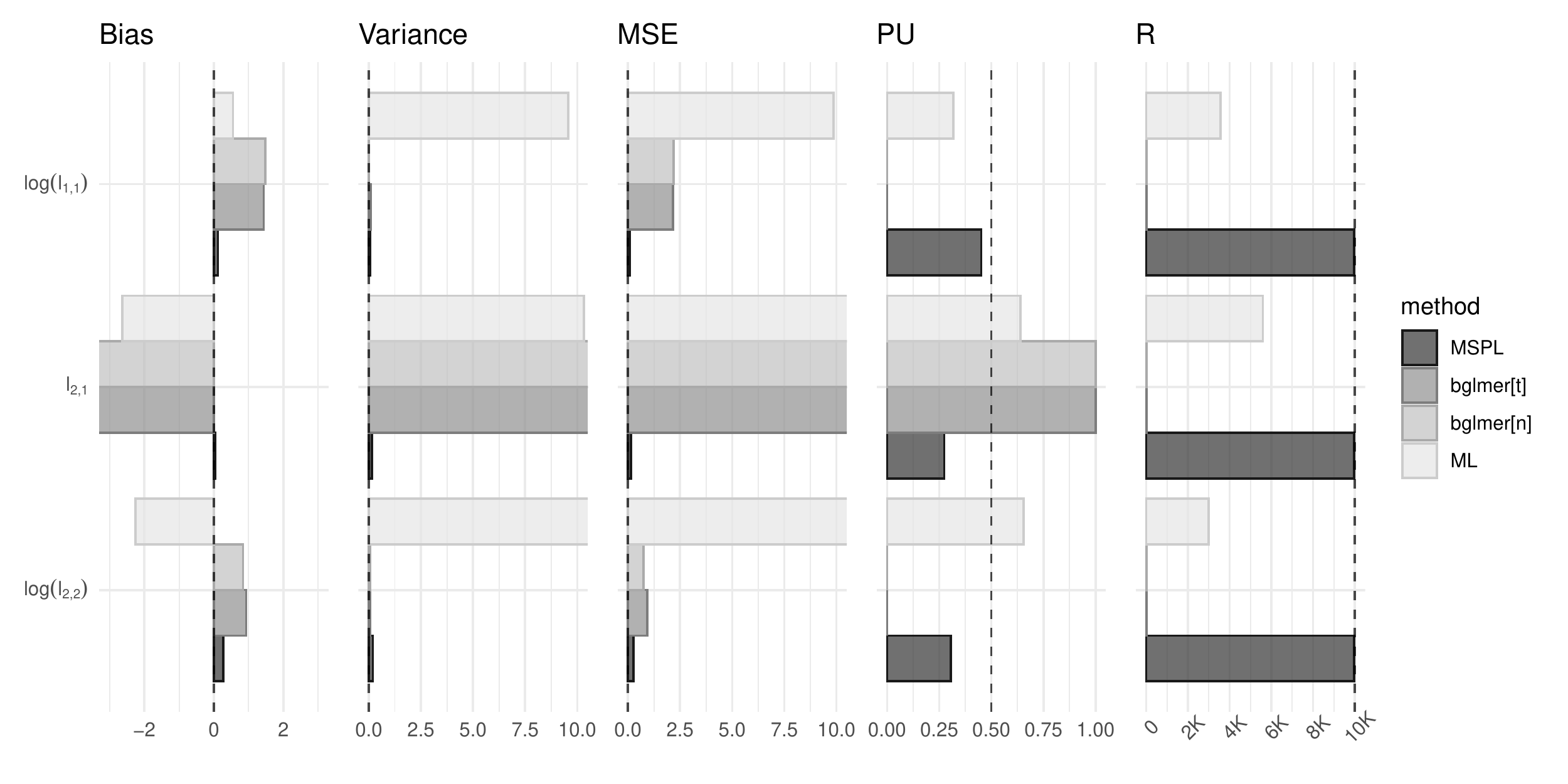}
		\end{center}
		\caption{Simulation based estimates of the bias, variance, mean squared-error (MSE), and the probability of underestimation (PU), for MSPL, bglmer, and ML estimators based on $10\ 000$ samples from~\eqref{eq:cond_inf_model} at the MSPL estimates in Table~\ref{tab:cond_inf}. Parameter estimates on the boundary are discarded from the calculation of the simulation summaries. The number of samples used for the calculation of summaries per parameter is given in the rightmost panel (R)}
		\label{fig:cond_inf_simul}
	\end{figure}
	
	To investigate the frequentist properties of the estimators on this data set, we repeat the simulation design of the Culcita data example from Section~\ref{sec:culcita_dat} for the conditional inference data at the MSPL estimate of Table~\ref{tab:cond_inf}. We point out the extremely low percentage of bglmer estimates without estimation issues that were used in the summary of Figure~\ref{fig:cond_inf_simul}. We note that the MSPL outperforms ML and bglmer, which incur substantial bias and variance due to their singular and infinite variance components estimates. Table~\ref{tab:sim2}, which shows the percentiles of the centered variance components estimates for each estimator, shows that ML and bglmer return heavily distorted variance components estimates, reflecting the fact that these estimators are unable to fully guard against degenerate variance components estimates. A comprehensive simulation summary for all parameters is given in Figure~\ref{fig:cond_inf_full_supp} and Table~\ref{tab:cond_inf_full_supp} of the Supplementary Material document. 
	
	\begin{table}[t]
		
		\caption{Percentiles of centered variance components estimates from simulating $10\ 000$ samples from model~(\ref{eq:cond_inf_model}) at the MSPL of Table~\ref{tab:cond_inf}}
		\label{tab:sim2}
		\begin{center}
			\begin{tabular}{lrccccccc}
				\toprule 
				&& \multicolumn{7}{c}{Percentiles} \\ \cmidrule{3-9} 
				&&$5\%$&$10\%$&$25\%$&$50\%$&$75\%$&$90\%$&$95\%$\\ \cmidrule{3-9} 
				\cmidrule{3-9} 
				& $\log l_{1,1}$ & -0.06 & -0.05 & -0.02 & 0.01 & 0.12 & 0.44 & 0.66 \\ 
				MSPL & $l_{2,1}$ & -0.66 & -0.27 & -0.02 & 0.12 & 0.21 & 0.35 & 0.45 \\ 
				& $\log l_{2,2}$ & -0.42 & -0.31 & -0.08 & 0.26 & 0.58 & 0.87 & 1.02 \\ 
				\cmidrule{3-9} 
				& $\log l_{1,1}$ & 1.02 & 1.07 & 1.28 & 1.40 & 1.63 & 1.82 & 1.85 \\ 
				bglmer[t] & $l_{2,1}$ & -33.53 & -31.71 & -26.77 & -2.19 & -1.30 & -0.69 & -0.67 \\ 
				& $\log l_{2,2}$ & 0.54 & 0.56 & 0.69 & 0.92 & 1.14 & 1.21 & 1.28 \\ 
				\cmidrule{3-9} 
				& $\log l_{1,1}$ & 1.40 & 1.40 & 1.41 & 1.45 & 1.47 & 1.57 & 1.64 \\ 
				bglmer[n] & $l_{2,1}$ & -29.76 & -3.72 & -3.39 & -1.74 & -1.53 & -0.93 & -0.63 \\ 
				& $\log l_{2,2}$ & 0.51 & 0.57 & 0.63 & 0.82 & 1.01 & 1.12 & 1.23 \\ 
				\cmidrule{3-9} 
				& $\log l_{1,1}$ & -2.75 & -2.05 & -0.91 & 1.21 & 2.36 & 2.62 & 2.66 \\ 
				ML & $l_{2,1}$ & -7.47 & -7.19 & -5.74 & -1.51 & 0.54 & 0.81 & 1.09 \\ 
				& $\log l_{2,2}$ & -6.62 & -4.54 & -3.22 & -1.77 & 0.39 & 0.85 & 1.04 \\ 
				\bottomrule 
			\end{tabular}
		\end{center}
	\end{table}
	
	\section{Concluding remarks}
	\label{sec:sum}
	This paper proposed the MSPL estimator for stable parameter estimation
	in mixed effects logistic regression. The method has been found, both
	theoretically and empirically, to have superior finite sample
	properties to the maximum penalized likelihood estimator proposed in
	\citet{chung+etal:2013}. We showed that penalizing the log-likelihood
	by scaled versions of the Jeffreys' prior for the model with no random
	effects and of a composition of the negative Huber loss gives
	estimates in the interior of the parameter space. Scaling the penalty
	function appropriately preserves the optimal ML asymptotics, namely
	consistency, asymptotic normality and Cram\'{e}r-Rao efficiency.
	
	We note that the conditions of Theorem~\ref{thm:soft_pen_cons}
	and Theorem~\ref{thm:asymp_norm_soft_pen} that are used for
	establishing the consistency and asymptotic normality of the MSPL
	estimator in Section~\ref{sec:asymptotics} are merely sufficient;
	there may be other sets of conditions that lead to the same results.
	
	While the MSPL is particularly relevant for mixed effects logistic
	regression, the concept is far more general and we expect it to be
	useful in other settings, for which degenerate ML estimates are known
	to occur, such as GLMMs with
	categorical or ordinal responses. The composite
	negative Huber loss penalty can be readily applied to other GLMMs, to
	prevent singular variance components estimates. We point out that the
	bound on the partial derivatives of the Jeffreys' prior in
	Theorem~\ref{thm:jeffrey_deriv_bound} for a logistic GLM extends to
	the cauchit link up to a constant; bounds for other link functions,
	like the probit and the complementary log-log are the subject of
	current work.
	
	\section{Supplementary Materials}
	The supplementary material to this paper is available at
	\url{https://github.com/psterzinger/softpen_supplementary}, and
	consists of the three folders ``Scripts'', ``Data'', ``Results'', and
	a Supplementary Material document, which provides further outputs and
	additional simulation studies. The ``Scripts'' directory contains all
	R scripts to reproduce the numerical analyses, simulations, graphics
	and tables in the main text and the Supplementary Material
	document. The ``Data'' directory contains the data used for the
	numerical examples in Sections~\ref{sec:culcita_dat} and~\ref{sec:ci},
	and the ``Results'' directory provides all results from the numerical
	experiments and analyses in the main text and the Supplementary
	Material. All numerical results are replicable in R version 4.2.2 (2022-10-31), and with the following packages: blme~1.0-5 \citep{chung+etal:2013}, doMC~1.3.8 \citep{doMC}, dplyr~1.0.10 \citep{dplyr}, 
	lme4~1.1-31 \citep{bates+etal:2015}, MASS~7.3-58.1 \citep{MASS}, Matrix~1.5-3 \citep{Matrix}, numDeriv~2016.8-1.1 \citep{gilbert+varadhan:2019}, optimx~2022-4.30 \citep{nash:2014}. A complete configuration is given in the supplementary material repository. 
	
	\appendix
	
	\section{Theorem~\ref{thm:nondeg}}
	
	\begin{theorem}
		\label{thm:nondeg}
		Let $\tilde{\bb L} \in \Re^{q \times q}$ be a real, lower triangular matrix with finite entries and strictly positive entries on its main diagonal. Then $\tilde{\bb \Sigma} = \tilde{\bb L}\tilde{\bb L}^\top$ is not degenerate. 
	\end{theorem}
	\begin{proof}
		Recall that a variance-covariance matrix is not degenerate if it is positive definite with finite entries, implying correlations away from one in absolute value. We prove each property in turn. To see that $\tilde{\bb \Sigma}$ is positive-definite, take any $\bb x\in \Re^q:\bb x\neq \bb 0_q$. Then by straightforward manipulations 
		\begin{equation}
		\label{eq:chol1}
		\bb x^\top\tilde{\bb \Sigma}\bb x = \bb x^\top\tilde{\bb L}\tilde{\bb L}^\top\bb x
		= (\tilde{\bb L}^\top\bb x)^\top\tilde{\bb L}^\top\bb x 
		= \langle \tilde{\bb L}^\top\bb x, \tilde{\bb L}^\top\bb x\rangle \geq 0
		\end{equation}
		where $\langle\cdot,\cdot\rangle$ denotes the standard
		Euclidean inner product. Hence $\tilde{\bb \Sigma}$ is
		positive semi-definite. Suppose that there is some
		$\bb x \in \Re^q$ such that $\bb x^\top \bb \Sigma \bb
		x=0$. Then by \eqref{eq:chol1},
		$\langle \tilde{\bb L}^\top\bb y, \tilde{\bb L}^\top\bb y
		\rangle = 0$ which holds if and only if
		$\tilde{\bb L}^\top\bb y=\bb 0_q$. Now since $\tilde{\bb L}$
		is lower triangular with strictly positive diagonal entries,
		it is full rank. But then
		$\bb y = \tilde{\bb L}^\top \bb x=\bb 0_q$ implies that
		$\bb x = \bb 0_q$ so that $\tilde{\bb \Sigma}$ is positive
		definite.  To prove that $\tilde{\bb \Sigma}$ has finite
		entries, note that
		$\tilde{\bb \Sigma}_{st} = \langle \tilde{\bb l}_s,\tilde{\bb
			l}_t \rangle$, where $\tilde{\bb l}_s$ is the $s$th row
		vector of $\tilde{\bb L}$. Since all elements of
		$\tilde{\bb l}_s,\tilde{\bb l}_t$ are finite, so is their
		inner product.  Finally, towards a contradiction, assume that
		$\tilde{\bb \Sigma}$ implies correlations of one in absolute
		value. Then there exist indices $s,t, s \neq t$ such that
		\[
		\left|\frac{\tilde{\bb \Sigma}_{st} }{\sqrt{ \tilde{\bb \Sigma}_{ss} \tilde{\bb \Sigma}_{tt}}}\right| =1 	\iff |\tilde{\bb \Sigma}_{st}| = \sqrt{ \tilde{\bb \Sigma}_{ss} \tilde{\bb \Sigma}_{tt}}  
		\iff |\langle \tilde{\bb l}_s,\tilde{\bb l}_t\rangle| = \|\tilde{\bb l}_s\|\|\tilde{\bb l}_t\|
		\]
		where $\|\bb x\|= \sqrt{\langle\bb x ,\bb x\rangle}$ is the induced inner product norm. It follows from the Cauchy-Schwarz inequality that the last equality holds if and only if $\tilde{\bb l}_s,\tilde{\bb l}_t$ are linearly dependent. Since $\tilde{\bb L}$ is lower triangular, this is only possible if $\tilde{\bb l}_s,\tilde{\bb l}_t$ have zeros in the same positions. But since all diagonal entries of $\tilde{\bb L}$ are strictly positive, this is not possible.
	\end{proof}
	
	\section{Consistency and asymptotic normality of MPL estimators}
	\label{sec:softpen}
	
	\subsection{Setup}
	
	Suppose that we observe the values $\by_1, \ldots, \by_k$ of a
	sequence of random vectors $\bY_1, \ldots, \bY_k$ with
	$\by_i = (y_{i1}, \ldots, y_{in_i})^\top \in \mathcal{Y} \subset
	\Re^{n_i}$, possibly with a sequence of covariate vectors
	$\bv_1, \ldots, \bv_k$, with
	$\bv_i = (v_{i1}, \ldots, v_{is})^\top \in \mathcal{X} \subset
	\Re^{s}$. Let $\bY = (\bY_1^\top, \ldots, \bY_k^\top)^\top$, and
	denote by $\bV$ the set of $\bv_1, \ldots, \bv_k$. Further, assume
	that the data generating process of $\bY$ conditional on $\bV$ has a
	density or probability mass function $f(\bY \mid \bV; \btheta)$,
	indexed by a parameter $\btheta \in \Theta \subset \Re^d$. Denote the
	parameter that identifies the conditional distribution of $\bY$ given
	$\bV$ by $\btnod \in \Theta$.
	
	Define the ML estimator as
	$\hat\btheta = \arg \max_{\btheta\in \Theta} \ell(\btheta)$, where
	$\ell(\btheta) = \log f(\bY \mid \bV; \btheta)$, and let
	$\tilde\btheta$ be the MPL estimator
	$\tilde\btheta = \arg\max_{\btheta \in \Theta}\{\ell(\btheta) +
	P(\btheta) \}$, where $P(\btheta)$ is an additive penalty to
	$\ell(\btheta)$ that may depend on $\bY$ and $\bV$. Consistency and
	asymptotic normality of the proposed MPL estimator follow readily from
	similar such results for ML estimators in \citet{ogden:2017} where the
	approximation error to the log-likelihood is an additive error
	term. In fact, the results presented in this section are a direct
	translation of the work in \citet{ogden:2017}, where the term
	``approximation error'' is replaced by ``penalty function'', and by
	allowing the rate of information accumulation to vary across the
	components of the parameter vector $\btheta$. Finally, let
	$\vnorm{\cdot}$ be some vector norm and $\mnorm{\cdot}$ be the
	corresponding operator norm.
	
	\subsection{Consistency}
	
	The consistency of $\bttilde$ can be established under the following
	regularity conditions on the log-likelihood gradient
	\citep[see][Chapter 5]{vaart:1998} and the penalty gradient.
	\begin{itemize}
		\item[A1] $\ell(\btheta)$ is differentiable with gradient
		$S(\btheta)$.
		\item[A2]
		$\underset{\btheta \in \Theta}{\sup} \; \vnorm{R_n^{-1}
			S(\btheta) - S_0(\btheta)} \overset{p}{\to}0$ for some
		deterministic function $S_0(\btheta)$, where $R_n$ is a
		diagonal matrix whose diagonal elements diverge to $+\infty$
		as $n$ grows.
		\item[A3] For all $\varepsilon>0$,
		$\underset{\btheta \in \Theta: \vnorm{\btheta-\btnod}\geq
			\varepsilon}{\inf} \vnorm{S_0(\btheta) }>0 =
		\vnorm{S_0(\btnod)}$.
		\item[A4] $P(\btheta)$ is differentiable with gradient $A(\btheta)$.
	\end{itemize}
	
	Define $\hat{\btheta}$ and $\tilde{\btheta}$ to be such that
	$S(\hat{\btheta}) = \b0_d$ and
	$S(\tilde{\btheta}) + A(\tilde{\btheta}) = \b0_d$,
	respectively. 

	\begin{theorem}[Consistency]
		\label{thm:soft_pen_cons}
		Suppose that A1-A4 hold and $\underset{\btheta \in \Theta}{\sup} \, \vnorm{R_n^{-1}A(\btheta)} = o_p(1)$. Then, $\bttilde \overset{p}{\to} \btnod$.
	\end{theorem}
	
	\begin{proof}
		The proof is analogous to the proof of \citet[Theorem
		1]{ogden:2017} and follows from
		\citet[Theorem~5.9]{vaart:1998} on the consistency of
		$M$-estimators. \citet[Theorem~5.9]{vaart:1998} states that
		under assumptions A2 and A3 about the log-likelihood
		gradient, if $\vnorm{R_n^{-1}S(\bttilde)} = o_p(1)$ then
		$\bttilde \overset{p}{\to} \btnod$. It holds that
		\begin{equation}
		\label{eq:soft_pen_cons_proof_1}
		\vnorm{ R_n^{-1}S(\bttilde) } = \vnorm{ R_n^{-1}\left\{S(\bttilde) + A(\bttilde)\right\} - R_n^{-1}A(\bttilde) }
		= \vnorm{ \b0_d- R_n^{-1}  A(\bttilde) } 
		= o_p(1) \, .
		\end{equation}
		The second equality follows from the definition of
		$\tilde\btheta$, and the last equality follows from the
		assumption that
		$\underset{\btheta \in \Theta}{\sup} \,
		\vnorm{R_n^{-1}A(\btheta)} = o_p(1)$.
		
	\end{proof}          
	
	\subsection{Asymptotic normality}
	
	The asymptotic normality of $\bttilde$ can be established under the
	following conditions
	
	\begin{itemize}
		\item[A5] $\ell(\btheta)$ is three times differentiable.
		\item[A6]
		$\sup_{\btheta \in \Theta}
		\mnorm{R_n^{-1/2}J(\btheta)R_n^{-1/2} -I(\btheta) }
		\overset{p}{\to} 0$ for some positive definite $O(1)$ matrix
		$I(\btheta)$, that is continuous in $\btheta$ in a
		neighbourhood around $\btnod$, where $R_n$ is a diagonal
		matrix whose diagonal elements diverge to $+\infty$ as $n$
		grows.
		\item[A7]
		$R_n^{1/2}(\hat\btheta-\btnod) \overset{d}{\to}
		\text{N}(0,I(\btnod)^{-1})$.
		\item[A8] $\bttilde \overset{p}{\to} \btnod$.
		\item[A9] $P(\btheta)$ is three times differentiable.
	\end{itemize}
	
	\begin{theorem}[Asymptotic Normality]
		\label{thm:asymp_norm_soft_pen}
		Suppose that A5-A9 hold, and $\underset{\btheta \in \Theta}{\sup} \, \vnorm{R_n^{-1/2}A(\btheta)} =~o_p(1)$. Then, $R_n^{1/2}(\bttilde-\btnod) \overset{d}{\to} \text{N}(0,I(\btnod)^{-1})$.
	\end{theorem}
	\begin{proof}
		We show that $R_n^{1/2}(\tilde{\bb \theta} -\hat{\bb \theta}) = o_p(1)$, which by A7, establishes the claim. Let $S(\btheta)$ be the gradient of $\ell(\btheta)$ and $J(\btheta) = -\nabla\nabla^\top \ell(\btheta)$ and consider a first-order Taylor expansion of $S(\btheta)$ around $\hat{\btheta}$. Then premultiplying by $R_n^{-1}$ gives
		\begin{equation}\label{eq:norm_taylor}
		R_n^{-1}S(\btheta) = R_n^{-1}S(\hat{\btheta}) + R_n^{-1} \nabla S(\btheta^*) (\btheta-\hat{\btheta}) = \b0_q -R_n^{-1}J(\btheta^*)(\btheta-\hat{\btheta})\, ,
		\end{equation}
		where $\btheta^*$ lies between $\btheta$ and $\hat{\btheta}$ and the second equality follows by the definition of $\hat{\btheta}$ and $J(\btheta)$. Hence
		\begin{equation}\label{eq:norm_taylor2}
		\b0_q = R_n^{-1}S(\bttilde) + R_n^{-1}A(\bttilde) = R_n^{-1} A(\bttilde) - R_n^{-1}J(\btheta^*)(\bttilde-\hat{\btheta})  \, , 
		\end{equation}
		where the first equality follows by the definition of $\bttilde$, and the second from substituting the right hand side of \eqref{eq:norm_taylor} for $S(\bttilde)$. Therefore 
		\begin{equation}
		R_n^{1/2}(\tilde{\bb \theta} -\hat{\bb \theta}) = [R_n^{-1/2}J(\btheta^*)R_n^{-1/2}]^{-1} [R_n^{-1/2}A(\bttilde)] \, . 
		\end{equation}
		Note that by A6, and an application of the continuous mapping theorem (see for example \citet[Theorem 2.1]{vaart:1998}), $\mnorm{[R_n^{-1/2}J(\btheta^*)R_n^{-1/2}]^{-1}} =~\Op{1}$. Hence,
		\begin{equation}
		\begin{aligned}
		\vnorm{R_n^{1/2}(\tilde{\bb \theta} -\hat{\bb \theta})} &= \vnorm{ [R_n^{-1/2}J(\btheta^*)R_n^{-1/2}]^{-1} [R_n^{-1/2}A(\bttilde)]} \\ 
		& \leq \mnorm{[R_n^{-1/2}J(\btheta^*)R_n^{-1/2}]^{-1} }\vnorm{R_n^{-1/2}A(\bttilde)} \\ 
		& \leq \mnorm{[R_n^{-1/2}J(\btheta^*)R_n^{-1/2}]^{-1} }\underset{\btheta \in \Theta}{\sup} \, \vnorm{R_n^{-1/2}A(\btheta)} \\ 
		& = o_p(1) \,,
		\end{aligned}
		\end{equation}
		where the first inequality follows from the definition of the operator norm, and the last line from the assumption of the Theorem. 
		Therefore $R_n^{1/2}(\tilde{\bb \theta} -\hat{\bb \theta}) = o_p(1)$ as required. 
	\end{proof}

	\subsection{Approximate likelihoods}
	
	We note that the large sample results for the MPL estimator derived
	here operate under the assumption that $\ell(\btheta)$ is the exact
	log-likelihood. If $\bar{\ell}(\btheta)$ is an approximation to the
	exact log-likelihood and $\bar{\btheta}$ is the maximizer of the
	penalized approximate likelihood, then consistency and asymptotic
	normality of $\bar{\btheta}$ can be established under extra conditions
	on $\bar{S}(\btheta)$, the gradient of the approximate likelihood
	$\bar{\ell}(\btheta)$. In particular, for consistency it is sufficient
	that
	$\underset{\btheta \in \Theta}{\sup} \, \vnorm{R_n^{-1}
		\left\{\bar{S}(\btheta)-S(\btheta) \right\}}=o_p(1)$, and for
	asymptotic normality it is sufficient that $\bar{\ell}(\btheta)$ is
	three-times differentiable and
	$\underset{\btheta \in \Theta}{\sup} \, \vnorm{R_n^{-1/2}
		\left\{\bar{S}(\btheta)-S(\btheta) \right\}}=o_p(1)$. In this
	instance, one can replace all occurrences of $A(\btheta)$ by
	$A(\btheta) + \bar{S}(\btheta)-S(\btheta)$ in the proofs of Theorems
	\ref{thm:soft_pen_cons} and \ref{thm:asymp_norm_soft_pen}. A simple
	application of the triangle inequality establishes that
	$\underset{\btheta \in \Theta}{\sup} \vnorm{ R_n^{-c}\left\{A(\btheta)
		+ \bar{S}(\btheta)-S(\btheta)\right\}} = o_p(1)$ as in the
	assumptions of Theorem~\ref{thm:soft_pen_cons} with $c = 1$, and
	Theorem~\ref{thm:asymp_norm_soft_pen} with $c = 1/2$, and thus the
	proofs apply. We refer the reader to \citet{ogden:2021} for
	approximation errors to the log-likelihood in clustered GLMMs using
	Laplace's method, \citet{ogden:2017} for approximation errors to the
	gradient of the log-likelihood with an example for an intercept-only
	Bernoulli-response GLMM, \citet{stringer:2022} for approximation
	errors to the log-likelihood in clustered GLMMs using adaptive
	Gauss-Hermite quadrature and \citet{jin+andersson:2020} for general
	approximation errors for adaptive Gauss-Hermite quadrature.
	
	\section{Bound on the gradient of the logarithm of the Jeffreys' prior}
	\begin{theorem}[Bound on the partial derivative of the logarithm of Jeffreys' prior]
		\label{thm:jeffrey_deriv_bound}
		Let $\bb X $ be the $n\times p$ full column rank matrix defined in Section~\ref{sec:bern_GLMMs}, and $\bb W$ a block-diagonal matrix with blocks $\bW_i$ as defined in Section~\ref{sec:composite_penalty} ($i=1,\ldots,k$). Then 
		\[        
		\left|\frac{\partial }{\partial  \beta_s}\log \det(\bb X^\top\bb W\bb X)\right| \leq p\underset{1\leq t\leq n}{\max} |x_{ts}|, 
		\]
		where $x_{ts}$ denotes the $t$th element in the $s$th column of $\bX$.
	\end{theorem}
	\begin{proof}
		We shall find it notationally convenient to neglect the block-structure of $\bX$ and refer to the $t$th element of the $s$th column of $\bX$ by $x_{ts}$ and define $\mu_t^{(f)}(\bbeta) = \exp(\eta_t^{(f)}(\bbeta))/(1+\exp(\eta_t^{(f)}(\bbeta)))$ for $\eta_t^{(f)}(\bbeta) = \bx_t^\top\bbeta$, and where $\bx_t^\top$ is the $t$th row of $\bX$. It is noted without proof that 
		\[
		\left|\frac{\partial }{\partial \beta_s}\log \textrm{det}(\bb X^\top\bb W\bb X)\right|  = \text{tr}\left( (\bb X^\top\bb W\bb X)^{-1}\bb X^\top\bb W\widetilde{\bb W}_s\bb X \right) \,, 
		\]
		where $\widetilde{\bb W}_s$ is a diagonal matrix with main-diagonal entries $\widetilde{w}^{(s)}_{t}= x_{ts}(1-2\mu_t^{(f)}(\bb\beta))$. 
		Now by the cyclical property of the trace, it follows that 
		\[
		\text{tr}\left( (\bb X^\top\bb W\bb X)^{-1}\bb X^\top\bb W\widetilde{\bb W}_s\bb X \right) =\text{tr}\left( \bb X(\bb X^\top\bb W\bb X)^{-1}\bb X^\top\bb W\widetilde{\bb W}_s \right) \,.
		\]
		For notational brevity, denote the projection matrix $\bb X(\bb X^\top\bb W\bb X)^{-1}\bb X^\top\bb W$ by $\bb P$. Since $\widetilde{\bb W}_s$ is a diagonal matrix, one gets that 
		\begin{align*} 
		\left| \text{tr}\left( \bb X(\bb X^\top\bb W\bb X)^{-1}\bb X^\top\bb W\widetilde{\bb W}_s \right) \right| &=  \left| \sum_{t=1}^{n}\widetilde{w}_t^{(s)}[\bb P]_{tt} \right| \\ 
		&\leq \sum_{t=1}^{n}\left|\widetilde{w}_t^{(s)}[\bb P]_{tt} \right| \\ 
		&= \sum_{t=1}^{n}\left|\widetilde{w}_t^{(s)}\right|[\bb P]_{tt}  \\ 
		&\leq \underset{1\leq t \leq n}{\max}\; \left|\widetilde{w}_t^{(s)}\right| \sum_{t=1}^{n}[\bb P]_{tt}  \\ 
		&=p \underset{1\leq t \leq n}{\max}\; \left|\widetilde{w}_t^{(s)}\right|\\ 
		&=p \underset{1\leq t \leq n}{\max}\; \left|x_{ts}(1-2\mu_t^{(f)}(\bb\beta))\right| \\ 
		&\leq p\underset{1\leq t \leq n}{\max}\;\left|x_{ts}\right|  \,.
		\end{align*}
		Here the second line is due to the triangle inequality. The third line follows by nonnegativity of the main-diagonal elements of $\bb P$. To see this, note that $\bb X^\top\bb W\bb X$ is positive definite as $\bb X$ has full column rank and $\bb W$ is a diagonal matrix with positive entries. It thus follows that $(\bb X^\top\bb W\bb X)^{-1}$ is positive definite. Hence for any $\bb y \in \Re^{n}, \|\bb y\|_2\neq 0$, $\bb y^\top\bb X(\bb X^\top\bb W\bb X)^{-1}\bb X \bb y = \tilde{\bb y}^\top(\bb X^\top\bb W\bb X)^{-1} \tilde{\bb y} \geq 0$, for $\tilde{\by} = \bX \by$, so that $\bb X(\bb X^\top\bb W\bb X)^{-1}\bb X$ is positive semi-definite. It is well known that the main diagonal entries of a positive semi-definite matrix are nonnegative. Hence, as $\bb W$ is a diagonal matrix with nonnegative diagonal entries it follows that the main diagonal entries of $\bb P$, which are the elementwise product of the diagonals of $\bb X(\bb X^\top\bb W\bb X)^{-1}\bb X$ and $\bb W$ are nonnegative. The fifth line follows since $\bb P$ is an idempotent matrix of rank $p$, and the fact that the trace of an idempotent matrix equals its rank \citep[Corollary 10.2.2]{harville:1998}. The fact that $\bb P$ has rank $p$ follows from the assumption that $\bb X$ has full column rank and since $\bb W$ is invertible for any $\bb\beta \in \Re^p,\bb X\in \Re^{n \times p}$ by construction and is a standard result in linear algebra (see for example \cite{magnus+neudecker:2019}, Chapter 1.7).  The last line follows since $\mu_t^{(f)}(\bb\beta) \in (0,1)$. 
		
	\end{proof}	
	
	\section*{Declarations} 
	
	For the purpose of open access, the author has applied a Creative Commons Attribution (CC BY) licence to any Author Accepted Manuscript version arising from this submission.
	
	\subsection*{Competing interests}
	The authors have no relevant financial or non-financial interests to disclose.

	\bibliographystyle{chicago}
	\bibliography{softpen}

\begin{thebibliography}{}

\bibitem[\protect\citeauthoryear{Bates, M{\"a}chler, Bolker, and Walker}{Bates
  et~al.}{2015}]{bates+etal:2015}
Bates, D., M.~M{\"a}chler, B.~Bolker, and S.~Walker (2015).
\newblock Fitting linear mixed-effects models using {lme4}.
\newblock {\em Journal of Statistical Software\/}~{\em 67\/}(1), 1--48.

\bibitem[\protect\citeauthoryear{Bates, Maechler, and Jagan}{Bates
  et~al.}{2022}]{Matrix}
Bates, D., M.~Maechler, and M.~Jagan (2022).
\newblock {\em Matrix: Sparse and Dense Matrix Classes and Methods}.
\newblock R package version 1.5-3.

\bibitem[\protect\citeauthoryear{Bolker}{Bolker}{2015}]{bolker:2015}
Bolker, B.~M. (2015).
\newblock Linear and generalized linear mixed models.
\newblock In G.~A. Fox, S.~{Negrete-Yankelevich}, and V.~J. Sosa (Eds.), {\em
  Ecological {{Statistics}}}, pp.\  309--333. {Oxford University Press}.

\bibitem[\protect\citeauthoryear{Bolker, Brooks, Clark, Geange, Poulsen,
  Stevens, and White}{Bolker et~al.}{2009}]{bolker+etal:2009}
Bolker, B.~M., M.~E. Brooks, C.~J. Clark, S.~W. Geange, J.~R. Poulsen, M.~H.~H.
  Stevens, and J.-S.~S. White (2009).
\newblock Generalized linear mixed models: a practical guide for ecology and
  evolution.
\newblock {\em Trends in ecology \& evolution\/}~{\em 24\/}(3), 127--135.

\bibitem[\protect\citeauthoryear{Breslow and Clayton}{Breslow and
  Clayton}{1993}]{breslow+clayton:1993}
Breslow, N.~E. and D.~G. Clayton (1993).
\newblock Approximate inference in generalized linear mixed models.
\newblock {\em Journal of the American statistical Association\/}~{\em
  88\/}(421), 9--25.

\bibitem[\protect\citeauthoryear{Browne and Draper}{Browne and
  Draper}{2006}]{browne+draper:2006}
Browne, W.~J. and D.~Draper (2006).
\newblock A comparison of bayesian and likelihood-based methods for fitting
  multilevel models.
\newblock {\em Bayesian analysis\/}~{\em 1\/}(3), 473--514.

\bibitem[\protect\citeauthoryear{Chung, Gelman, Rabe-Hesketh, Liu, and
  Dorie}{Chung et~al.}{2015}]{chung+etal:2015}
Chung, Y., A.~Gelman, S.~Rabe-Hesketh, J.~Liu, and V.~Dorie (2015).
\newblock Weakly informative prior for point estimation of covariance matrices
  in hierarchical models.
\newblock {\em Journal of Educational and Behavioral Statistics\/}~{\em
  40\/}(2), 136--157.

\bibitem[\protect\citeauthoryear{Chung, Rabe-Hesketh, Dorie, Gelman, and
  Liu}{Chung et~al.}{2013}]{chung+etal:2013}
Chung, Y., S.~Rabe-Hesketh, V.~Dorie, A.~Gelman, and J.~Liu (2013).
\newblock A nondegenerate penalized likelihood estimator for variance
  parameters in multilevel models.
\newblock {\em Psychometrika\/}~{\em 78\/}(4), 685--709.

\bibitem[\protect\citeauthoryear{Gelman, Jakulin, Pittau, and Su}{Gelman
  et~al.}{2008}]{gelman+etal:2008}
Gelman, A., A.~Jakulin, M.~G. Pittau, and Y.-S. Su (2008).
\newblock {A weakly informative default prior distribution for logistic and
  other regression models}.
\newblock {\em The Annals of Applied Statistics\/}~{\em 2\/}(4), 1360 -- 1383.

\bibitem[\protect\citeauthoryear{Gilbert and Varadhan}{Gilbert and
  Varadhan}{2019}]{gilbert+varadhan:2019}
Gilbert, P. and R.~Varadhan (2019).
\newblock {\em numDeriv: Accurate Numerical Derivatives}.
\newblock R package version 2016.8-1.1.

\bibitem[\protect\citeauthoryear{Harville}{Harville}{1998}]{harville:1998}
Harville, D.~A. (1998).
\newblock Matrix algebra from a statistician's perspective.

\bibitem[\protect\citeauthoryear{Heinze and Schemper}{Heinze and
  Schemper}{2002}]{heinze+schemper:2002}
Heinze, G. and M.~Schemper (2002).
\newblock A solution to the problem of separation in logistic regression.
\newblock {\em Statistics in Medicine\/}~{\em 21\/}(16), 2409--2419.

\bibitem[\protect\citeauthoryear{Jiang}{Jiang}{2017}]{jiang:2017}
Jiang, J. (2017).
\newblock {\em Asymptotic analysis of mixed effects models: theory,
  applications, and open problems}.
\newblock Chapman and Hall/CRC.

\bibitem[\protect\citeauthoryear{Jin and Andersson}{Jin and
  Andersson}{2020}]{jin+andersson:2020}
Jin, S. and B.~Andersson (2020).
\newblock A note on the accuracy of adaptive gauss--hermite quadrature.
\newblock {\em Biometrika\/}~{\em 107\/}(3), 737--744.

\bibitem[\protect\citeauthoryear{Konis}{Konis}{2007}]{konis:2017}
Konis, K. (2007).
\newblock {\em Linear programming algorithms for detecting separated data in
  binary logistic regression models}.
\newblock Ph.\ D. thesis, University of Oxford.

\bibitem[\protect\citeauthoryear{Kosmidis and Firth}{Kosmidis and
  Firth}{2021}]{kosmidis+firth:2021}
Kosmidis, I. and D.~Firth (2021).
\newblock Jeffreys-prior penalty, finiteness and shrinkage in binomial-response
  generalized linear models.
\newblock {\em Biometrika\/}~{\em 108\/}(1), 71--82.

\bibitem[\protect\citeauthoryear{Kosmidis and Schumacher}{Kosmidis and
  Schumacher}{2021}]{kosmidis+schumacher:2021}
Kosmidis, I. and D.~Schumacher (2021).
\newblock {\em detectseparation: Detect and Check for Separation and Infinite
  Maximum Likelihood Estimates}.
\newblock R package version 0.2.

\bibitem[\protect\citeauthoryear{Liu and Pierce}{Liu and
  Pierce}{1994}]{liu+pierce:1994}
Liu, Q. and D.~A. Pierce (1994).
\newblock A note on {G}auss-{H}ermite quadrature.
\newblock {\em Biometrika\/}~{\em 81\/}(3), 624--629.

\bibitem[\protect\citeauthoryear{Magnus and Neudecker}{Magnus and
  Neudecker}{2019}]{magnus+neudecker:2019}
Magnus, J.~R. and H.~Neudecker (2019).
\newblock {\em Matrix differential calculus with applications in statistics and
  econometrics}.
\newblock John Wiley \& Sons.

\bibitem[\protect\citeauthoryear{McCullagh and Nelder}{McCullagh and
  Nelder}{1989}]{mccullagh+nelder:1989}
McCullagh, P. and J.~A. Nelder (1989).
\newblock {\em Generalized Linear Models\/} (2nd ed.).
\newblock {Boca Raton}: {Chapman \& Hall/CRC}.

\bibitem[\protect\citeauthoryear{McCulloch}{McCulloch}{1997}]{mcculloch:1997}
McCulloch, C.~E. (1997, March).
\newblock Maximum {{Likelihood Algorithms}} for {{Generalized Linear Mixed
  Models}}.
\newblock {\em Journal of the American Statistical Association\/}~{\em
  92\/}(437), 162--170.

\bibitem[\protect\citeauthoryear{McCulloch, Searle, and Neuhaus}{McCulloch
  et~al.}{2008}]{mcculloch+etal:2008}
McCulloch, C.~E., S.~R. Searle, and J.~M. Neuhaus (2008).
\newblock {\em Generalized, linear, and mixed models}, Volume~2.
\newblock John Wiley \& Sons.

\bibitem[\protect\citeauthoryear{McKeon, Stier, McIlroy, and Bolker}{McKeon
  et~al.}{2012}]{mckeon+etal:2012}
McKeon, C.~S., A.~C. Stier, S.~E. McIlroy, and B.~M. Bolker (2012).
\newblock Multiple defender effects: synergistic coral defense by mutualist
  crustaceans.
\newblock {\em Oecologia\/}~{\em 169\/}(4), 1095--1103.

\bibitem[\protect\citeauthoryear{Nash}{Nash}{2014}]{nash:2014}
Nash, J.~C. (2014).
\newblock On best practice optimization methods in {R}.
\newblock {\em Journal of Statistical Software\/}~{\em 60\/}(2), 1--14.

\bibitem[\protect\citeauthoryear{Ogden}{Ogden}{2017}]{ogden:2017}
Ogden, H. (2017).
\newblock On asymptotic validity of naive inference with an approximate
  likelihood.
\newblock {\em Biometrika\/}~{\em 104\/}(1), 153--164.

\bibitem[\protect\citeauthoryear{Ogden}{Ogden}{2021}]{ogden:2021}
Ogden, H. (2021).
\newblock On the error in laplace approximations of high-dimensional integrals.
\newblock {\em Stat\/}~{\em 10\/}(1), e380.

\bibitem[\protect\citeauthoryear{Pasch, Bolker, and Phelps}{Pasch
  et~al.}{2013}]{pasch+etal:2013}
Pasch, B., B.~M. Bolker, and S.~M. Phelps (2013).
\newblock Interspecific dominance via vocal interactions mediates altitudinal
  zonation in neotropical singing mice.
\newblock {\em The American Naturalist\/}~{\em 182\/}(5), E161--E173.

\bibitem[\protect\citeauthoryear{Pinheiro and Bates}{Pinheiro and
  Bates}{1995}]{pinheiro+bates:1995}
Pinheiro, J.~C. and D.~M. Bates (1995).
\newblock Approximations to the {Log}-{Likelihood} {Function} in the
  {Nonlinear} {Mixed}-{Effects} {Model}.
\newblock {\em Journal of Computational and Graphical Statistics\/}~{\em
  4\/}(1), 12--35.

\bibitem[\protect\citeauthoryear{Pinheiro and Chao}{Pinheiro and
  Chao}{2006}]{pinheiro+chao:2006}
Pinheiro, J.~C. and E.~C. Chao (2006).
\newblock Efficient laplacian and adaptive gaussian quadrature algorithms for
  multilevel generalized linear mixed models.
\newblock {\em Journal of Computational and Graphical Statistics\/}~{\em
  15\/}(1), 58--81.

\bibitem[\protect\citeauthoryear{{R Core Team}}{{R Core Team}}{2022}]{R}
{R Core Team} (2022).
\newblock {\em R: A Language and Environment for Statistical Computing}.
\newblock Vienna, Austria: R Foundation for Statistical Computing.

\bibitem[\protect\citeauthoryear{Raudenbush, Yang, and Yosef}{Raudenbush
  et~al.}{2000}]{raudenbush+etal:2000}
Raudenbush, S.~W., M.-L. Yang, and M.~Yosef (2000).
\newblock Maximum likelihood for generalized linear models with nested random
  effects via high-order, multivariate laplace approximation.
\newblock {\em Journal of computational and Graphical Statistics\/}~{\em
  9\/}(1), 141--157.

\bibitem[\protect\citeauthoryear{{Revolution Analytics} and Weston}{{Revolution
  Analytics} and Weston}{2022}]{doMC}
{Revolution Analytics} and S.~Weston (2022).
\newblock {\em doMC: Foreach Parallel Adaptor for 'parallel'}.
\newblock R package version 1.3.8.

\bibitem[\protect\citeauthoryear{Rigollet}{Rigollet}{2015}]{rigollet:2015}
Rigollet, P. (2015).
\newblock {\em 18.S997 High-Dimensional Statistics}.
\newblock https://ocw.mit.edu: Massachusetts Institute of Technology: MIT
  OpenCourseWare.

\bibitem[\protect\citeauthoryear{Rodriguez and Goldman}{Rodriguez and
  Goldman}{1995}]{rodriguez+goldman1995}
Rodriguez, G. and N.~Goldman (1995).
\newblock An assessment of estimation procedures for multilevel models with
  binary responses.
\newblock {\em Journal of the Royal Statistical Society: Series A (Statistics
  in Society)\/}~{\em 158\/}(1), 73--89.

\bibitem[\protect\citeauthoryear{Schall}{Schall}{1991}]{schall:1991}
Schall, R. (1991).
\newblock Estimation in generalized linear models with random effects.
\newblock {\em Biometrika\/}~{\em 78\/}(4), 719--727.

\bibitem[\protect\citeauthoryear{Singmann, Klauer, and Beller}{Singmann
  et~al.}{2016}]{singmann+etal:2016}
Singmann, H., K.~C. Klauer, and S.~Beller (2016).
\newblock Probabilistic conditional reasoning: Disentangling form and content
  with the dual-source model.
\newblock {\em Cognitive Psychology\/}~{\em 88}, 61--87.

\bibitem[\protect\citeauthoryear{Stringer and Bilodeau}{Stringer and
  Bilodeau}{2022}]{stringer:2022}
Stringer, A. and B.~Bilodeau (2022).
\newblock Fitting generalized linear mixed models using adaptive quadrature.

\bibitem[\protect\citeauthoryear{Vaart}{Vaart}{1998}]{vaart:1998}
Vaart, A. W. v.~d. (1998).
\newblock {\em Asymptotic Statistics}.
\newblock Cambridge Series in Statistical and Probabilistic Mathematics.
  Cambridge University Press.

\bibitem[\protect\citeauthoryear{Venables and Ripley}{Venables and
  Ripley}{2002}]{MASS}
Venables, W.~N. and B.~D. Ripley (2002).
\newblock {\em Modern Applied Statistics with S\/} (Fourth ed.).
\newblock New York: Springer.
\newblock ISBN 0-387-95457-0.

\bibitem[\protect\citeauthoryear{Wickham, François, Henry, and
  Müller}{Wickham et~al.}{2022}]{dplyr}
Wickham, H., R.~François, L.~Henry, and K.~Müller (2022).
\newblock {\em dplyr: A Grammar of Data Manipulation}.
\newblock R package version 1.0.9.

\bibitem[\protect\citeauthoryear{Wolfinger and O'connell}{Wolfinger and
  O'connell}{1993}]{wolfinger+oconnel:1993}
Wolfinger, R. and M.~O'connell (1993).
\newblock Generalized linear mixed models a pseudo-likelihood approach.
\newblock {\em Journal of statistical Computation and Simulation\/}~{\em
  48\/}(3-4), 233--243.

\bibitem[\protect\citeauthoryear{Zehna}{Zehna}{1966}]{zehna:1966}
Zehna, P.~W. (1966).
\newblock Invariance of {Maximum} {Likelihood} {Estimators}.
\newblock {\em The Annals of Mathematical Statistics\/}~{\em 37\/}(3),
  744--744.

\bibitem[\protect\citeauthoryear{Zhao, Staudenmayer, Coull, and Wand}{Zhao
  et~al.}{2006}]{zhao+etal:2006}
Zhao, Y., J.~Staudenmayer, B.~A. Coull, and M.~P. Wand (2006).
\newblock {General Design Bayesian Generalized Linear Mixed Models}.
\newblock {\em Statistical Science\/}~{\em 21\/}(1), 35 -- 51.

\end{thebibliography}



\section*{Supplementary material}

# Supplementary Material document for Maximum softly-penalized likelihood for mixed effects logistic regression

Philipp Sterzinger
`philipp.sterzinger@warwick.ac.uk`
and
Ioannis Kosmidis
`ioannis.kosmidis@warwick.ac.uk`

Department of Statistics, University of Warwick
Coventry CV4 7AL, UK

February 2, 2023

## S1 Supplementary Material

All labels for the sections, equations, tables, figures and so on in the current document have been prefixed by "S" (e.g. Section S1, equation (S1), Table S3, etc). The supplementary material for *Maximum softly-penalized likelihood for mixed effects logistic regression* contains:

i) ML, MSPL and bglmer (Chung et al., 2013) estimates for the full Culcita data set of McKeon et al. (2012) as referenced in Section 3 of the main text (see Section S2).

ii) A summary of the simulation study of Section 8 of the main paper (see Section S3).

iii) Further simulations on synthetic data (see Section S4).

This document and the R scripts and datasets to reproduce the results in the main text and the present document are available at `https://github.com/psterzinger/softpen_supplementary`.

## S2 Culcita data

Table S1: MSPL, ML and bglmer estimates from the full Culcita dataset of (McKeon et al., 2012), using "none" as reference category. Asymptotic standard errors based on the negative Hessian of the approximate log-likelihood are given in parentheses

|            | ML[BFGS] | ML[CG] | bglmer[t] | bglmer[n] | MSPL   |
|------------|----------|--------|-----------|-----------|--------|
| $\beta_0$  | 5.01     | 5.01   | 3.57      | 3.55      | 4.23   |
|            | (1.80)   | (1.80) | (1.53)    | (1.52)    | (1.53) |
| $\beta_2$  | $-3.75$  | $-3.75$| $-2.21$   | $-2.19$   | $-3.40$|
|            | (1.46)   | (1.46) | (1.11)    | (1.11)    | (1.35) |
| $\beta_3$  | $-4.36$  | $-4.36$| $-2.68$   | $-2.66$   | $-3.93$|
|            | (1.55)   | (1.55) | (1.16)    | (1.15)    | (1.42) |
| $\beta_4$  | $-5.55$  | $-5.55$| $-3.65$   | $-3.62$   | $-4.95$|
|            | (1.72)   | (1.72) | (1.27)    | (1.27)    | (1.56) |
| $\log \sigma$ | 1.26  | 1.26   | 1.24      | 1.24      | 1.13   |
|            | (0.37)   | (0.37) | (0.39)    | (0.39)    | (0.36) |

## S3 Conditional inference data

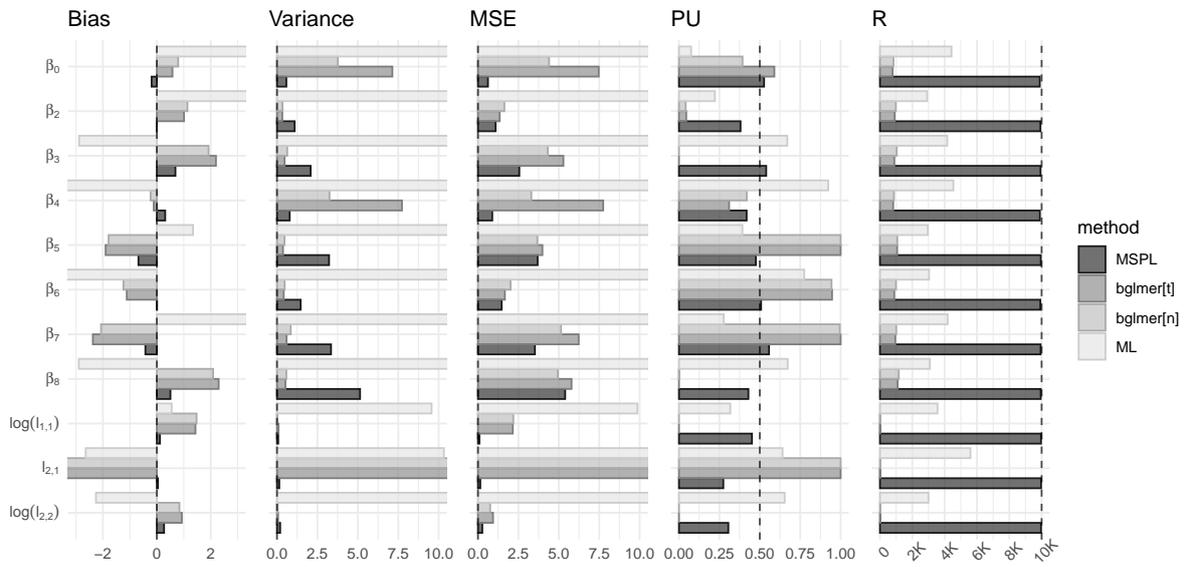

Figure S1: Full simulation results from the simulation study in Section 8 of the main text

Table S2: Percentiles of centred parameter estimates from the simulation study in Section 8 of the main text

|  |  | Percentiles | | | | | | |
|---|---|---|---|---|---|---|---|---|
|  |  | 5% | 10% | 25% | 50% | 75% | 90% | 95% |
| MSPL | $\beta_0$ | -2.17 | -2.03 | -0.01 | -0.00 | 0.02 | 0.13 | 0.27 |
|  | $\beta_2$ | -2.18 | -2.17 | -0.00 | 0.00 | 0.00 | 2.17 | 2.18 |
|  | $\beta_3$ | -1.07 | -0.87 | -0.65 | -0.00 | 2.17 | 2.18 | 2.18 |
|  | $\beta_4$ | -0.60 | -0.46 | -0.21 | 0.11 | 0.52 | 1.88 | 2.31 |
|  | $\beta_5$ | -4.35 | -2.18 | -2.17 | 0.00 | 0.66 | 1.08 | 2.18 |
|  | $\beta_6$ | -2.30 | -1.56 | -0.51 | -0.00 | 0.52 | 1.62 | 2.30 |
|  | $\beta_7$ | -3.27 | -2.73 | -1.84 | -0.23 | 0.78 | 1.68 | 2.82 |
|  | $\beta_8$ | -3.24 | -2.35 | -0.96 | 0.41 | 2.05 | 3.39 | 4.31 |
|  | $\log l_{1,1}$ | -0.06 | -0.05 | -0.02 | 0.01 | 0.12 | 0.44 | 0.66 |
|  | $l_{2,1}$ | -0.66 | -0.27 | -0.02 | 0.12 | 0.21 | 0.35 | 0.45 |
|  | $\log l_{2,2}$ | -0.42 | -0.31 | -0.08 | 0.26 | 0.58 | 0.87 | 1.02 |
| bglmer[t] | $\beta_0$ | -2.01 | -1.81 | -1.43 | -0.31 | 1.67 | 5.22 | 5.69 |
|  | $\beta_2$ | 0.03 | 0.21 | 0.56 | 1.06 | 1.45 | 1.73 | 1.88 |
|  | $\beta_3$ | 1.31 | 1.47 | 1.77 | 2.10 | 2.41 | 3.02 | 3.70 |
|  | $\beta_4$ | -5.60 | -5.10 | -0.97 | 0.91 | 1.90 | 2.42 | 2.68 |
|  | $\beta_5$ | -3.02 | -2.69 | -2.20 | -1.82 | -1.51 | -1.23 | -1.05 |
|  | $\beta_6$ | -2.10 | -1.90 | -1.55 | -1.15 | -0.68 | -0.31 | 0.03 |
|  | $\beta_7$ | -3.61 | -3.25 | -2.87 | -2.38 | -1.86 | -1.42 | -1.10 |
|  | $\beta_8$ | 1.12 | 1.43 | 1.81 | 2.27 | 2.73 | 3.18 | 3.47 |
|  | $\log l_{1,1}$ | 1.02 | 1.07 | 1.28 | 1.40 | 1.63 | 1.82 | 1.85 |
|  | $l_{2,1}$ | -33.53 | -31.71 | -26.77 | -2.19 | -1.30 | -0.69 | -0.67 |
|  | $\log l_{2,2}$ | 0.54 | 0.56 | 0.69 | 0.92 | 1.14 | 1.21 | 1.28 |
| bglmer[n] | $\beta_0$ | -1.25 | -1.04 | -0.58 | 0.31 | 0.91 | 4.39 | 4.80 |
|  | $\beta_2$ | 0.04 | 0.25 | 0.76 | 1.33 | 1.55 | 1.73 | 1.81 |
|  | $\beta_3$ | 0.86 | 1.02 | 1.37 | 1.76 | 2.49 | 2.87 | 3.31 |
|  | $\beta_4$ | -4.09 | -3.59 | -0.54 | 0.24 | 1.02 | 1.45 | 1.77 |
|  | $\beta_5$ | -3.05 | -2.72 | -2.18 | -1.68 | -1.34 | -0.99 | -0.81 |
|  | $\beta_6$ | -2.20 | -2.02 | -1.73 | -1.36 | -0.81 | -0.26 | 0.06 |
|  | $\beta_7$ | -3.71 | -3.33 | -2.64 | -2.01 | -1.43 | -0.93 | -0.68 |
|  | $\beta_8$ | 0.94 | 1.20 | 1.57 | 2.02 | 2.52 | 3.05 | 3.48 |
|  | $\log l_{1,1}$ | 1.40 | 1.40 | 1.41 | 1.45 | 1.47 | 1.57 | 1.64 |
|  | $l_{2,1}$ | -29.76 | -3.72 | -3.39 | -1.74 | -1.53 | -0.93 | -0.63 |
|  | $\log l_{2,2}$ | 0.51 | 0.57 | 0.63 | 0.82 | 1.01 | 1.12 | 1.23 |
| ML | $\beta_0$ | -1.11 | 1.58 | 5.84 | 9.60 | 11.76 | 13.58 | 14.87 |
|  | $\beta_2$ | -8.42 | -7.29 | 0.76 | 4.96 | 7.59 | 10.71 | 12.75 |
|  | $\beta_3$ | -10.06 | -8.74 | -7.04 | -5.40 | 2.16 | 5.85 | 9.43 |
|  | $\beta_4$ | -14.59 | -13.43 | -11.64 | -9.51 | -5.65 | -1.42 | 1.28 |
|  | $\beta_5$ | -10.07 | -6.58 | -1.34 | 0.74 | 2.89 | 8.63 | 16.11 |
|  | $\beta_6$ | -12.92 | -10.61 | -7.41 | -4.88 | -1.04 | 7.29 | 8.40 |
|  | $\beta_7$ | -6.77 | -5.02 | -0.91 | 5.61 | 7.60 | 10.38 | 13.80 |
|  | $\beta_8$ | -19.81 | -11.80 | -5.30 | -1.21 | 0.80 | 3.99 | 8.77 |
|  | $\log l_{1,1}$ | -2.75 | -2.05 | -0.91 | 1.21 | 2.36 | 2.62 | 2.66 |
|  | $l_{2,1}$ | -7.47 | -7.19 | -5.74 | -1.51 | 0.54 | 0.81 | 1.09 |
|  | $\log l_{2,2}$ | -6.62 | -4.54 | -3.22 | -1.77 | 0.39 | 0.85 | 1.04 |

# S4 Further Simulations

## S4.1 Simulation 1: Extreme fixed effects

This section presents a simulation that seeks to provoke degenerate fixed effects estimates through a strong dependence of the responses on the particular fixed effects – a phenomenon known to occur in standard logistic regression (Albert and Anderson, 1984; Kosmidis and Firth, 2021).

For this, we simulate from a mixed effects logistic regression with univariate random effects and logistic link function as follows. For five clusters $i = 1, \ldots, 5$ and within cluster observations $j = 1, \ldots, n$, $n \in \{50, 100, 200\}$, we draw an i.i.d. vector of fixed effects covariates $\boldsymbol{x}_{ij} = (x_{i1}, x_{i2}, x_{i3}, x_{i4}, x_{i5})$ where $X_{i1} = 1, X_{i2} \sim \text{N}(0,1), X_{i3} \sim \text{Ber}\left(\frac{1}{2}\right), X_{i4} \sim \text{Ber}\left(\frac{1}{4}\right)$, and $X_{i5} \sim \exp(1)$. The fixed effect covariates are drawn once and held fixed over the simulation. To control the degree of dependence of the responses on a particular fixed effect covariate, the parameter of fixed effects is set as $\boldsymbol{\beta} = (1, -0.5, \lambda, 0.25, -1)$, where $\lambda$ takes integer values from $-10$ to $10$. For each specification of $n, \lambda$, we draw 100 samples from the model

$$Y_{ij} \mid u_i \sim \text{Bernoulli}(\mu_{ij}) \quad \text{with} \quad g(\mu_{ij}) = \eta_{ij} = \boldsymbol{x}_{ij}^\top \boldsymbol{\beta} + u_i,$$
$$u_i \sim \text{N}(0, 9) \quad (i = 1, \ldots, 5; j = 1, \ldots, n), \quad n \in \{50, 100, 200\} \tag{S1}$$

The random effects dispersion parameter $\sigma = 3$ was chosen as to avoid estimation issues associated with small random effects. We estimate the parameters using our proposed MSPL with the penalties given in Section 5 of the main text, ML and bglmer from the `blme` R (R Core Team, 2022) package (Chung et al., 2013) with a normal and t prior for the fixed effects and a gamma prior for the random effects variance. We approximate the log-likelihood with a 20-point adaptive Gauss-Hermite quadrature approximation. For ML and MSPL, we optimize the approximate log-likelihood using the optimization methods "CG", "BFGS", "nlminb" and "L-BFGS-B" from the `optimx` R package (Nash, 2014) and report the best fit. Both bglmer specifications use the default bglmer optimization settings.

Table S3 shows the number of estimates per specification which resulted in an degenerate parameter estimate. We considered an estimate degenerate, if it is larger than 50 in absolute value, the gradient is larger than 0.001 in absolute value, or if the estimated asymptotic standard errors are larger than 40. Figure S2 shows the dispersion of the estimates $\beta_3$ around the true value (indicated by dashed horizontal line) per specification for all estimation methods. For presentability, the y-axis has been cropped to omit overly extreme estimates. For the ML, 672 estimates are cropped, for bglmer[t] one observation is not shown, while all estimates are shown for MSPL. Note that the boxplots for the ML do not depict the empirical distribution of the ML estimate of $\beta_3$ over different samples as theoretically infinite estimates assume finite values in estimation due to numerical precision limitations and resulting premature declaration of convergence. We see from Table S3 that the MSPL gives the most stable parameter estimation, with one out of 6000 samples exhibiting estimation issues due to failed convergence. The ML on the other hand becomes highly problematic for large values of $\beta_3$ and even returns degenerate estimates for moderate values of $\beta_3$ with non-negligible frequency. The bglmer estimates, even though they penalize fixed effects more harshly, as can be seen by the shrinkage-induced bias of the bglmer estimates of $\beta_3$ in Figure S2, they encounter estimation issues substantially more often than the MSPL estimates.

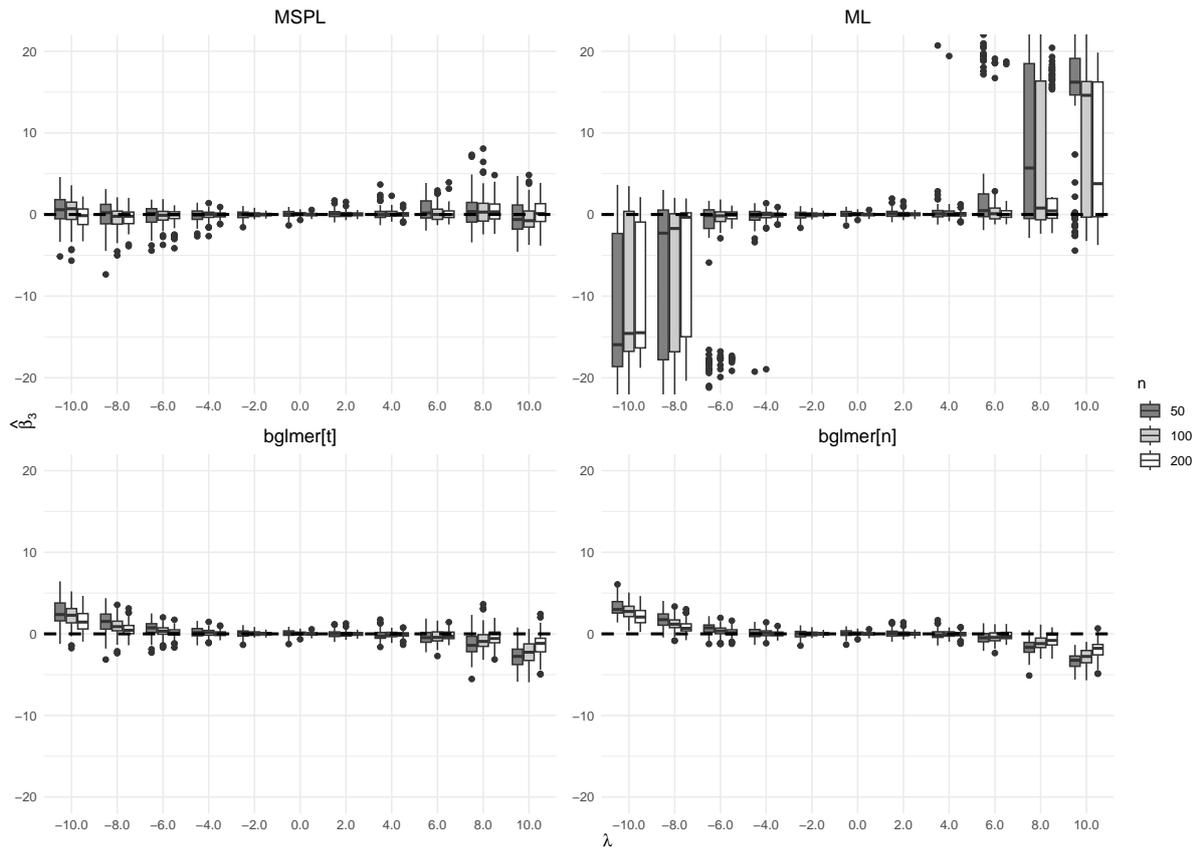

Figure S2: Centered estimation output of $\beta_3$ from Simulation 1

Table S3: Percentage of degenerate estimates from Simulation 1

|  |  | \multicolumn{11}{c}{$\lambda$} |
|---|---|---|---|---|---|---|---|---|---|---|---|---|
|  |  | -10 | -8 | -6 | -4 | -2 | 0 | 2 | 4 | 6 | 8 | 10 |
| MSPL | n=50 | 0 | 0 | 0 | 0 | 0 | 0 | 0 | 0 | 0 | 0 | 0 |
|  | n=100 | 0 | 0 | 0 | 0 | 0 | 0 | 0 | 0 | 0 | **1** | 0 |
|  | n=200 | 0 | 0 | 0 | 0 | 0 | 0 | 0 | 0 | 0 | 0 | 0 |
| ML | n=50 | **73** | **46** | **17** | **2** | 0 | 0 | 0 | **1** | **15** | **49** | **80** |
|  | n=100 | **63** | **45** | **10** | **1** | 0 | 0 | 0 | **1** | **5** | **31** | **61** |
|  | n=200 | **65** | **26** | **8** | 0 | 0 | 0 | 0 | 0 | **2** | **24** | **49** |
| bglmer[t] | n=50 | **1** | 0 | 0 | 0 | 0 | 0 | 0 | 0 | 0 | 0 | 0 |
|  | n=100 | 0 | 0 | 0 | **1** | 0 | 0 | 0 | 0 | 0 | 0 | 0 |
|  | n=200 | 0 | 0 | **1** | 0 | **1** | 0 | 0 | 0 | 0 | 0 | 0 |
| bglmer[n] | n=50 | **1** | 0 | 0 | 0 | 0 | **1** | 0 | 0 | 0 | 0 | 0 |
|  | n=100 | 0 | 0 | **1** | 0 | **1** | 0 | 0 | 0 | **1** | 0 | 0 |
|  | n=200 | 0 | 0 | 0 | 0 | 0 | 0 | 0 | 0 | 0 | 0 | 0 |

## S4.2 Simulation 2: Extreme random effects variance

In this simulation, we seek to provoke degenerate random effects variance estimates, that is random effects variance estimates that are either zero or infinite. One of the peculiarities of Bernoulli-response (or Binomial-response) mixed effects models is that there can be separation of the observations with respect to the random effects covariates. Analogously to separation in logistic regression models (Albert and Anderson, 1984), where covariate constellations such that the responses can be separated by a hyperplane spanned by the covariate column vectors, lead to infinite parameter estimates, it is known that certain constellations of random effects covariates can lead to data separation and consequently degenerate random effects estimates (see for example Sauter and Held (2016) or the discussion on https://stats.stackexchange.com/questions/44755). We consider a simple simulation to provoke such data configurations by simulating from a mixed effects logistic regression with univariate random effects and vary the dependence of the responses on the grouping variable by controlling the random effects variance parameter.

For five clusters $i = 1, \ldots, 5$ and within cluster observations $j = 1, \ldots, n$, $n \in \{50, 100, 200\}$, we draw an i.i.d. vector of fixed effects covariates $\boldsymbol{x}_{ij} = (x_{i1}, x_{i2}, x_{i3}, x_{i4}, x_{i5})$ where $X_{i1} = 1$, $X_{i2} \sim \mathrm{N}(0,1)$, $X_{i3} \sim \mathrm{Ber}\left(\frac{1}{2}\right)$, $X_{i4} \sim \mathrm{Ber}\left(\frac{1}{4}\right)$, and $X_{i5} \sim \exp(1)$. The fixed effect covariates are drawn once and held fixed over the simulation. Likewise, the fixed effects $\boldsymbol{\beta} = (1, -0.5, 0.5, 0.25, -1)$ are held fixed over the simulation, while $\lambda = \log \sigma$ is varied over the integer values from $-5$ to $2$. For each specification of $n, \lambda$, we draw 100 samples from the model

$$\begin{aligned}
Y_{ij} \mid u_i &\sim \mathrm{Bernoulli}(\mu_{ij}) \quad \text{with} \quad g(\mu_{ij}) = \eta_{ij} = \boldsymbol{x}_{ij}^\top \boldsymbol{\beta} + u_i, \\
u_i &\sim \mathrm{N}(0, \exp(\lambda)^2) \quad (i = 1, \ldots, 5; j = 1, \ldots, n), \quad n \in \{50, 100, 200\}
\end{aligned} \tag{S2}$$

We estimate the parameters using our proposed MSPL with the penalties given in Section 6 of the main text, ML and bglmer from the `blme` R package (Chung et al., 2013) with a normal and t prior for the fixed effects and a gamma prior for the random effects variance. We approximate the log-likelihood with a 20-point adaptive Gauss-Hermite quadrature approximation. For ML and MSPL, we optimize the approximate log-likelihood using the optimization methods "CG", "BFGS", "nlminb" and "L-BFGS-B" from the `optimx` R package (Nash, 2014) and report the best fit. Both bglmer specifications use the default bglmer optimization settings. Figure S3 shows the dispersion of the estimates for $\log \sigma$ around the true value (indicated by dashed horizontal line), for each estimation method and specification of $\lambda, n$. For the ML and bglmer estimates, these boxplots do not approximate the distribution of the maximum likelihood estimator as, owed to numerical precision limitations, parameter estimates which ought to be infinite or zero are not estimated as such so that the point masses at the boundaries of the parameter space are missing. For bglmer[t], 7 estimates are not shown and for bglmer[n], 6 estimates are not shown due to failed estimation. While both the MSPL and bglmer shrink negative estimates of $\log \sigma$ towards zero, the shrinkage induced by bglmer is considerably stronger, as can be seen by the absolute amount of shrinkage and the smaller dispersion of the estimates. Moreover, we see that for larger values of $\log \sigma$ bglmer is unable to guard against infinite estimates of the random effects variance.

Table S4 shows the number of estimates for $\log \sigma$ per specification which resulted in an degenerate random effects variance estimate. We considered an estimate degenerate, if it is larger than $\log(50)$ or smaller than $\log(0.01)$, the gradient is larger than $0.001$ in absolute value, or if the estimated asymptotic standard errors are larger than 40. We see that MSPL is the most stable estimation routine and that both bglmer and ML exhibit estimation degeneracies frequently for both small and large values of the random effects variance.

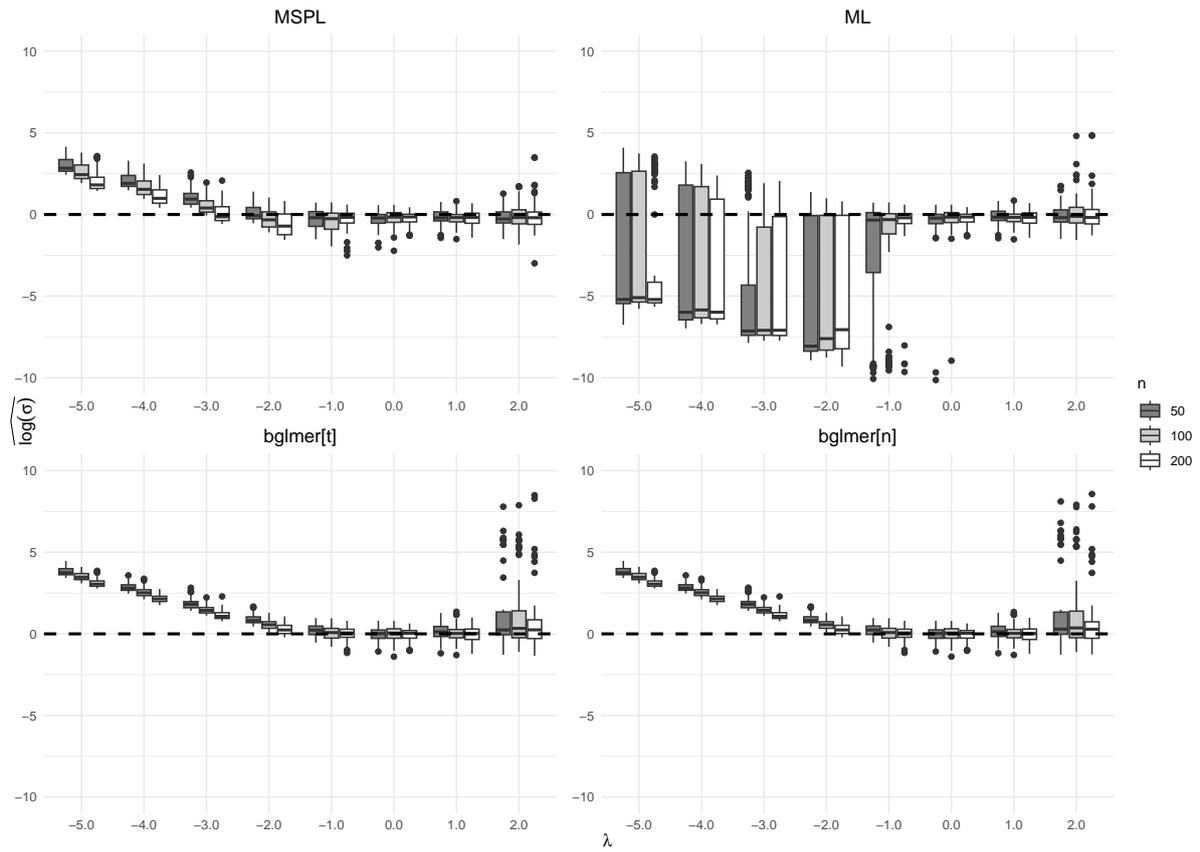

Figure S3: Centered estimation output of $\log \sigma$ from Simulation 2

Table S4: Percentage of degenerate estimates from Simulation 2

|  |  | $\lambda$ | | | | | | | |
| --- | --- | --- | --- | --- | --- | --- | --- | --- | --- |
|  |  | -5 | -4 | -3 | -2 | -1 | 0 | 1 | 2 |
| MSPL | n=50 | 0 | 0 | 0 | 0 | 0 | 0 | 0 | **1** |
|  | n=100 | 0 | 0 | 0 | 0 | 0 | 0 | 0 | 0 |
|  | n=200 | 0 | 0 | 0 | 0 | 0 | 0 | 0 | **2** |
| ML | n=50 | **73** | **70** | **76** | **70** | **25** | **2** | 0 | 0 |
|  | n=100 | **68** | **68** | **75** | **58** | **16** | **1** | 0 | **7** |
|  | n=200 | **78** | **68** | **69** | **52** | **4** | 0 | 0 | **5** |
| bglmer[t] | n=50 | 0 | 0 | 0 | 0 | 0 | 0 | 0 | **12** |
|  | n=100 | 0 | 0 | 0 | 0 | 0 | 0 | 0 | **18** |
|  | n=200 | 0 | 0 | 0 | 0 | 0 | 0 | 0 | **11** |
| bglmer[n] | n=50 | 0 | 0 | 0 | 0 | 0 | 0 | 0 | **12** |
|  | n=100 | 0 | 0 | 0 | 0 | 0 | 0 | 0 | **17** |
|  | n=200 | 0 | 0 | 0 | 0 | 0 | 0 | 0 | **11** |


\end{document}